%% file: main.tex
\documentclass[journal,twocolumn]{IEEEtran}
\include{Definition}
\usepackage{enumitem}
\usepackage{marginnote}

\begin{document}
        \setlength{\abovedisplayskip}{4.5pt}
	\setlength{\belowdisplayskip}{4.5pt}
	\title{Performance Analysis and Power Allocation for Massive MIMO ISAC Systems}
	\author{Nhan~Thanh~Nguyen, \IEEEmembership{Member, IEEE},
		Van-Dinh~Nguyen, \IEEEmembership{Senior Member, IEEE},
		Hieu~V.~Nguyen, \IEEEmembership{Member, IEEE},\\
		Hien~Quoc~Ngo, \IEEEmembership{Fellow, IEEE},
		A.~Lee~Swindlehurst, \IEEEmembership{Fellow, IEEE},
		and Markku Juntti, \IEEEmembership{Fellow, IEEE}
        \thanks{This work was supported in part by the Research Council of Finland through the 6G Flagship program (grant number 369116), project DIRECTION (grant number 354901), and project FunISAC (grant 359093), by CHIST-ERA through the project PASSIONATE  (grant number 359817), by Business Finland, Keysight, MediaTek, Siemens, Ekahau, and Verkotan via the 6GLearn project, and by HORIZON-JU-SNS-2023 project INSTINCT. An earlier version of this paper was presented in part at the IEEE Radar Conference, 2024.} \thanks{Nhan Thanh Nguyen and Markku Juntti are with Centre for Wireless Communications, University of Oulu, P.O.Box 4500, FI-90014, Finland (e-mail:\{nhan.nguyen, markku.juntti\}@oulu.fi). Van-Dinh Nguyen is with the College of Engineering and Computer Science and also with the Center for Environmental Intelligence, VinUniversity, Vinhomes Ocean Park, Hanoi 100000, Vietnam (e-mail: dinh.nv2@vinuni.edu.vn). Hieu Van Nguyen is with the Faculty of Electronic and Telecommunication Engineering, The University of Danang, University of Science and Technology, Da Nang 50000, Vietnam (email: nvhieu@dut.udn.vn). Hien Quoc Ngo is with the School of Electronics, Electrical Engineering and Computer Science, Queen’s University Belfast, Belfast BT7 1NN, United Kingdom (email: hien.ngo@qub.ac.uk). A. Lee Swindlehurst is with the Center for Pervasive Communications and Computing, University of California, Irvine, CA 92697, US (email: swindle@uci.edu).}
	}
	\maketitle
	
	\begin{abstract}
		Integrated sensing and communications (ISAC) is envisioned as a key feature in future wireless communications networks. Its integration with massive multiple-input-multiple-output (MIMO) techniques promises to leverage substantial spatial beamforming gains for both functionalities. In this work, we consider a massive MIMO-ISAC system employing a uniform planar array with zero-forcing and maximum-ratio downlink transmission schemes combined with monostatic radar-type sensing. Our focus lies on deriving closed form expressions for the achievable communications rate and the Cramér--Rao lower bound (CRLB), which serve as performance metrics for communications and sensing operations, respectively. The expressions enable us to investigate important operational characteristics of massive MIMO-ISAC, including the mutual effects of communications and sensing as well as the advantages stemming from using a very large antenna array for each functionality. Furthermore, we devise a power allocation strategy based on successive convex approximation to maximize the communications rate while guaranteeing the CRLB constraints and transmit power budget. Extensive numerical results are presented to validate our theoretical analyses and demonstrate the efficiency of the proposed power allocation approach.
	\end{abstract}
	
	\IEEEpeerreviewmaketitle
	
	\section{Introduction}
	Future wireless communications technologies such as evolving 6G systems will be required to meet increasingly demanding objectives. These objectives encompass supporting high-throughput, low-latency communications while ensuring high energy efficiency. Additionally, 6G is anticipated to offer sensing and cognition capabilities~\cite{giordani2020toward}, enabled by massive MIMO transceivers designed for dual communications and radar purposes~\cite{zhang2021overview}. This emerging concept of unifying communications and sensing is commonly referred to as dual-functional radar-communications, joint communications and sensing~~\cite{zhang2018multibeam}, or integrated sensing and communications (ISAC)~\cite{ouyang2022performance, liu2022integrated}. For consistency, we adopt the term ISAC throughout this work. Our focus lies in transceiver designs for joint design of communications and radar sensing operations in massive multiple-input-multiple-output (MIMO) systems, wherein a substantial number of antennas are deployed at the base station (BS).

	\subsection{Related Works}
	Various forms of ISAC systems have been proposed in the literature, categorized based on their design focus into radar-centric, communications-centric, and joint designs ~\cite{ma2020joint,zhang2021overview}. The radar-centric approach builds upon existing radar technologies, extending their functionality to incorporate communications capabilities. This is often realized by integrating digital messages into radar waveforms via index modulation~\cite{huang2020majorcom, ma2021spatial, ma2021frac} or by modulating the radar side lobes~\cite{Hassanien2016Dual}. Communications-centric approaches, on the other hand, usually employ conventional communications signals/waveforms for probing the environment~\cite{Kumari2018IEEE80211ad}, though typically with restricted sensing functionalities. Joint ISAC designs are our focus in this paper, which optimize the transceiver configuration to strike a balance between communications and sensing.

	Recent literature has seen an increasing emphasis on transmit beamforming designs for ISAC systems~\cite{liu2018mu, liu2020joint, johnston2022mimo, liu2022joint, liu2022transmit, pritzker2022transmit, wang2023transmit}. Various design metrics were investigated for the newly integrated sensing function, such as the beampattern~\cite{liu2018mu, liu2020joint, johnston2022mimo, wang2023transmit}, signal-to-cluster-plus-noise ratio (SCNR)~\cite{choi2024joint, zhang2023isac, wang2023optimizing}, Cramér--Rao lower bound (CRLB)~\cite{liu2021cramer, song2023intelligent, zhu2023cramer, ren2022fundamental, song2023cram, song2023cramer, wang2021joint}, among others~\cite{johnston2022mimo}. Specifically, in~\cite{liu2018mu, liu2020joint, johnston2022mimo}, the transmit beamformers were designed to minimize the sensing beampattern mismatch constrained by the communications signal-to-interference-plus-noise ratio (SINR) and transmit power budget. The works~\cite{nguyen2023multiuser, nguyen2023joint} aimed at maximizing the communications throughput while constraining the beampattern error. In the designs in~\cite{choi2024joint, zhang2023isac}, the SCNR is cast as a beamforming design constraint, while both the SINRs and SCNR were considered in a weighted design objective in~\cite{wang2023optimizing} to enhance the fairness between communications users and the sensing target.
	
	The ISAC designs relying on the beampattern or SCNRs ensure that the sensing target(s) are covered within the main lobes of the radar transmitter. However, these metrics do not directly address the processing of radar echo signals for detection or estimation. In contrast, the CRLB serves as a reliable lower bound for the target's parameter estimation accuracy. It was thus employed as the sensing metric in~\cite{liu2021cramer, song2023intelligent, zhu2023cramer, ren2022fundamental, song2023cram, song2023cramer, wang2021joint} for transmit beamforming designs for point and/or extended target scenarios. Note that due to the typical form of the CRLB, the constraints and objectives related to it are often converted into semidefinite forms using the Schur complement. This transformation enables them to be addressed via semi-definite relaxation (SDR)~\cite{liu2021cramer, song2023cram, song2023cramer, zhu2023cramer, song2023intelligent}. 
	
	Most of the aforementioned works investigated ISAC operations in MIMO scenarios~\cite{li2016optimum, liu2018mu, liu2017robust, liu2022transmit}. However, conventional small-sized antenna arrays may not ensure high spatial beamforming gains. In contrast, massive MIMO technology with very large arrays offers superior spectral and energy efficiency (SE/EE) for communications systems~\cite{marzetta2010noncooperative, larsson2014massive, ngo2013energy, fortunati2023fundamental}. It has been shown in~\cite{buzzi2019using} that by leveraging a massive MIMO radar BS, communications and radar systems can coexist with little mutual interference. Temiz \textit{et al.}~\cite{temiz2021dual} proposed a joint uplink massive MIMO communications and orthogonal frequency-division multiplexing (OFDM) radar sensing architecture, using zero-forcing (ZF) and ordered successive interference cancellation receivers to eliminate the inter-user and radar interference during communications symbol detection. In~\cite{qi2022hybrid, wang2022partially, liyanaarachchi2021joint, barneto2021beamformer}, the radar performance of massive MIMO systems is optimized under constraints on communications performance. On the other hand, the communications rate of massive MIMO-ISAC systems is maximized in~\cite{nguyen2023jointssp, nguyen2023joint, nguyen2023multiuser} subject to constraints on the radar performance. In \cite{gao2022integrated, zhang2024integrated}, efficient channel and target parameter estimation approaches were developed for millimeter-wave massive MIMO ISAC systems, relying on compressed sampling \cite{gao2022integrated} and tensor frameworks \cite{zhang2024integrated}. {Liao \textit{et al.}\cite{liao2023power} derived closed form expressions for the communications SINR and the sensing mainlobe-to-average-sidelobe ratio (MASR). They also proposed a power allocation algorithm to minimize the total transmit power while satisfying predefined SINR and MASR thresholds. Topal \textit{et al.}\cite{topal2024multi} focused on a multi-target sensing scenario and introduced a beamforming scheme aimed at minimizing the total CRLBs for direction-of-arrival (DoA) estimates of all targets.}
	
	\subsection{Motivation and Contributions}
	
	In the aforementioned works on massive MIMO ISAC and/or CRLB optimizations, system designs are typically performed at the rate that the small-scale fading changes. Furthermore, most existing CRLB-based ISAC designs rely on the highly complex SDR method and focus on estimation of only a single angle with uniform linear arrays (ULAs). However, these approaches may not be suitable for massive MIMO ISAC scenarios where very large uniform planar arrays (UPAs) are often deployed in practice~\cite{lu2020omnidirectional}. Owing to its many degrees of freedom, massive MIMO can provide good beamforming gains with simple linear beamforming methods such as maximum-ratio transmission (MRT) and ZF and additional power allocation over long time intervals~\cite{ngo2013massive, liao2023power}. 
	
	In this paper, we consider a mono-static downlink multiuser massive MIMO-ISAC system employing a large UPA transmitter and MRT or ZF precoders. {For the sensing function, we focus on the tracking mode and assume that the searching phase has been completed, i.e., the presence of a target successfully detected \cite{liu2021cramer, song2023intelligent}.} To characterize the joint communications and sensing operations of the system, we investigate the sum rate for communications and the CRLB for sensing. We note that existing results on the CRLB and techniques for addressing the CRLB constraints are not directly applicable to our design. This is because we assume a UPA and perform power allocations based on the large-scale fading parameters. Our main contributions are as follows:
	\begin{itemize}
		\item We derive closed form expressions for the achievable rate and the CRLB for the target's azimuth and elevation angles. These metrics are expressed as functions of the large-scale fading parameters and are unified for both the MRT and ZF schemes. The approach significantly simplifies the subsequent power allocation procedure.
		\item Based on the closed form rate and CRLB expressions, we show that the sensing objective increases the beamforming uncertainty and the inter-user interference of the communications channels for both MRT and ZF. However, these drawbacks can be mitigated by deploying a very large number of antennas. We also show that the sensing performance can be improved by increasing the number of antennas.
		\item We perform power allocation to maximize the communications rate while constraining the CRLB. The formulated problem is nonconvex, but it can be solved with low complexity in our proposed algorithm by leveraging the successive convex approximation (SCA) approach. We also provide an efficient initial solution for faster convergence.
		\item Finally, we present numerical results to verify our theoretical findings and demonstrate the performance of the power allocation scheme.
	\end{itemize}

{We summarize the contributions and novelties of this paper compared to existing studies on massive MIMO ISAC systems in Table \ref{tab_comparison}. While SINR and achievable rate are commonly used to evaluate communications performance, the sensing performance metrics vary across the compared studies. Notably, works such as \cite{liao2023power, buzzi2019using, temiz2021dual} provide closed form expressions for the communications SINR and achievable rate. However, most of these studies evaluate sensing performance using instantaneous metrics, requiring beamforming designs that adapt to the rapid changes of small-scale fading channels. In contrast, our study is unique in deriving closed form CRLBs for the estimation of both azimuth and elevation angles. Furthermore, although closed form achievable rate expressions have been derived for massive MIMO communications-only systems in works such as \cite{ngo2013energy, xin2015area, pirzadeh2018spectral, liu2016energy, zhang2016spectral, kamga2016spectral, bjornson2017massive, van2017massive}, these results are not applicable to massive MIMO ISAC systems because ISAC waveforms differ fundamentally from those used in communications-only systems.}

\begin{table*}[h!]\small
\centering
\caption{{Comparison of contributions and novelties of closely related works on massive MIMO ISAC}}
\label{tab_comparison}
\resizebox{\textwidth}{!}{%
\begin{tabular}{|m{0.75cm}|m{2cm}|m{3cm}|m{11cm}|}
\hline
\multicolumn{1}{|c|}{\textbf{Ref.}} & 
\multicolumn{1}{c|}{\textbf{Comm. metric}} & 
\multicolumn{1}{c|}{\textbf{Sensing metric}} & 
\multicolumn{1}{c|}{\textbf{Main contributions}} \\ 
\hline 
\hline

\cite{buzzi2019using} & closed form achievable rate & Instantaneous probability of detection & \begin{itemize}[label=\textbullet, nolistsep, leftmargin=5pt,
            before*={\mbox{}\vspace{-\baselineskip}}, after*={\mbox{}\vspace{-\baselineskip}}]
    \item Propose using massive array at the BS to perform downlink communications and sensing/surveillance of the surrounding environment through radar scanning.
    \item Show that communications and radar scanning functions can coexist with minimal interference thanks to the use of large antenna arrays.
\end{itemize}  \\ 
\hline

\cite{temiz2021dual} & closed form achievable rate & Instantaneous radar channel estimation mean squared error and radar image SINR & \begin{itemize}[label=\textbullet, nolistsep, leftmargin=5pt,
            before*={\mbox{}\vspace{-\baselineskip}}, after*={\mbox{}\vspace{-\baselineskip}}]
    \item Propose a dual-function radar-communications system to communicate with multiple uplink UEs while sensing targets in range using the same OFDM subcarriers.
    \item Derive expression for closed form communications rate with perfect and imperfect CSI.
    \item Evaluate radar detection and channel estimation accuracy.
\end{itemize}  \\ 
\hline

\cite{liao2023power} & closed form SINR & closed form MASR & \begin{itemize}[label=\textbullet, nolistsep, leftmargin=5pt,
            before*={\mbox{}\vspace{-\baselineskip}}, after*={\mbox{}\vspace{-\baselineskip}}]
    \item Derive closed form expressions for SINR and MASR for communications and sensing performance with MRT and ZF precoding.
    \item Calculate power allocation to minimize total transmit power under communications SINR and sensing MASR constraints.
\end{itemize}  \\ 
\hline

\cite{topal2024multi} & Instantaneous SINR & Instantaneous CRLBs for target DoA estimates & \begin{itemize}[label=\textbullet, nolistsep, leftmargin=5pt,
            before*={\mbox{}\vspace{-\baselineskip}}, after*={\mbox{}\vspace{-\baselineskip}}]
    \item Minimize the CRLB trace for target DoA estimates subject to communications SINR and power constraints, employ regularized ZF precoding for sensing.
    \item Consider multi-target sensing scenario.
\end{itemize}  \\ 
\hline

This work & closed form achievable rate & closed form CRLBs for target azimuth and elevation angles estimates & \begin{itemize}[label=\textbullet, nolistsep, leftmargin=5pt,
            before*={\mbox{}\vspace{-\baselineskip}}, after*={\mbox{}\vspace{-\baselineskip}}]
    \item Derive a unified closed form expression for communications rate for both MRT and ZF precoding  \textit{(not found in \cite{topal2024multi})}.
    \item Derive closed form sensing CRLBs for target azimuth and elevation angle estimates \textit{(not found in \cite{liao2023power, temiz2021dual, buzzi2019using, topal2024multi})}.
    \item (i) Show that sensing degrades communications rate by increasing beamforming uncertainty and inter-user interference, (ii) demonstrate how this effect can be mitigated by deploying a large number of antennas, and (iii) CRLBs approach zero when the number of antennas goes to infinity \textit{(not found in \cite{liao2023power, temiz2021dual, buzzi2019using, topal2024multi})}.
    \item Derive power allocation to maximize communications sum rate subject to sensing CRLB and total power constraints \textit{(different from \cite{liao2023power, temiz2021dual, buzzi2019using, topal2024multi})}.
\end{itemize}  \\ 
\hline
\end{tabular}%
}
\end{table*}

	\subsection{Paper Organization and Notation}
	The rest of the paper is organized as follows. In Section \ref{sec_system_model}, we present the communications and sensing signal models and the general linear beamformers. In Section \ref{sec_perf_analysis}, we derive closed form expressions for the system communications rate and sensing CRLB, followed by important operational characteristics of the considered massive MIMO ISAC system. Section \ref{sec_opt} details the proposed power allocation method. Numerical results are given and discussed in Section \ref{sec_sim}. Finally, Section \ref{sec_conclusion} concludes the paper. 
	
	Throughout the paper, scalars, vectors, and matrices are denoted by lower-case, boldface lower-case, and boldface upper-case letters, respectively; $(\cdot)^*$, $(\cdot)^\T$, $(\cdot)^\H$, and $\tr{\cdot}$ denote the conjugate, the transpose,  the conjugate transpose, and the trace operators, respectively; $\abs{\cdot}$ and $\norm{\cdot}$  respectively denote the modulus of a complex number and the Euclidean norm of a vector. The notation $\dot{\va}_o$ represents the derivative of $\va$ with respect to $o$, i.e., $\dot{\va}_o = \frac{\partial \va}{\partial o}$, $\mathcal{CN}(\mu, \sigma^2)$ denotes a complex normal distribution with mean $\mathbf{\mu}$ and variance $\sigma^2$, and $\mean{\cdot}$ denotes the expected value of a random variable. 
	
	\section{Signal Model}
	\label{sec_system_model}

	We consider mono-static massive MIMO ISAC downlink data transmission and passive target sensing, including a BS, $K$ single-antenna communications users, and a sensed target. We assume that the BS is equipped with a UPA consisting of $\Nt$ transmit and $\Nr$ receive antennas, with $\Nt \gg 1$ and $\Nr \gg 1$. {At the BS, the $\Nt$ antennas simultaneously transmit data signals to the users and probing signals to the target located at a specific angle of interest. This angle is assumed to have been determined during the radar scanning or searching phase \cite{liu2021cramer, song2023cramer, song2023intelligent, zhu2023cramer, song2023cram}.} The echo from the sensed target is then processed by the $\Nr$ receive antennas at the BS.
	
	\subsection{Communications Model}
	\subsubsection{Communications Signal Model}
	\begin{figure}[t]
		\centering
		\includegraphics[scale=0.14]{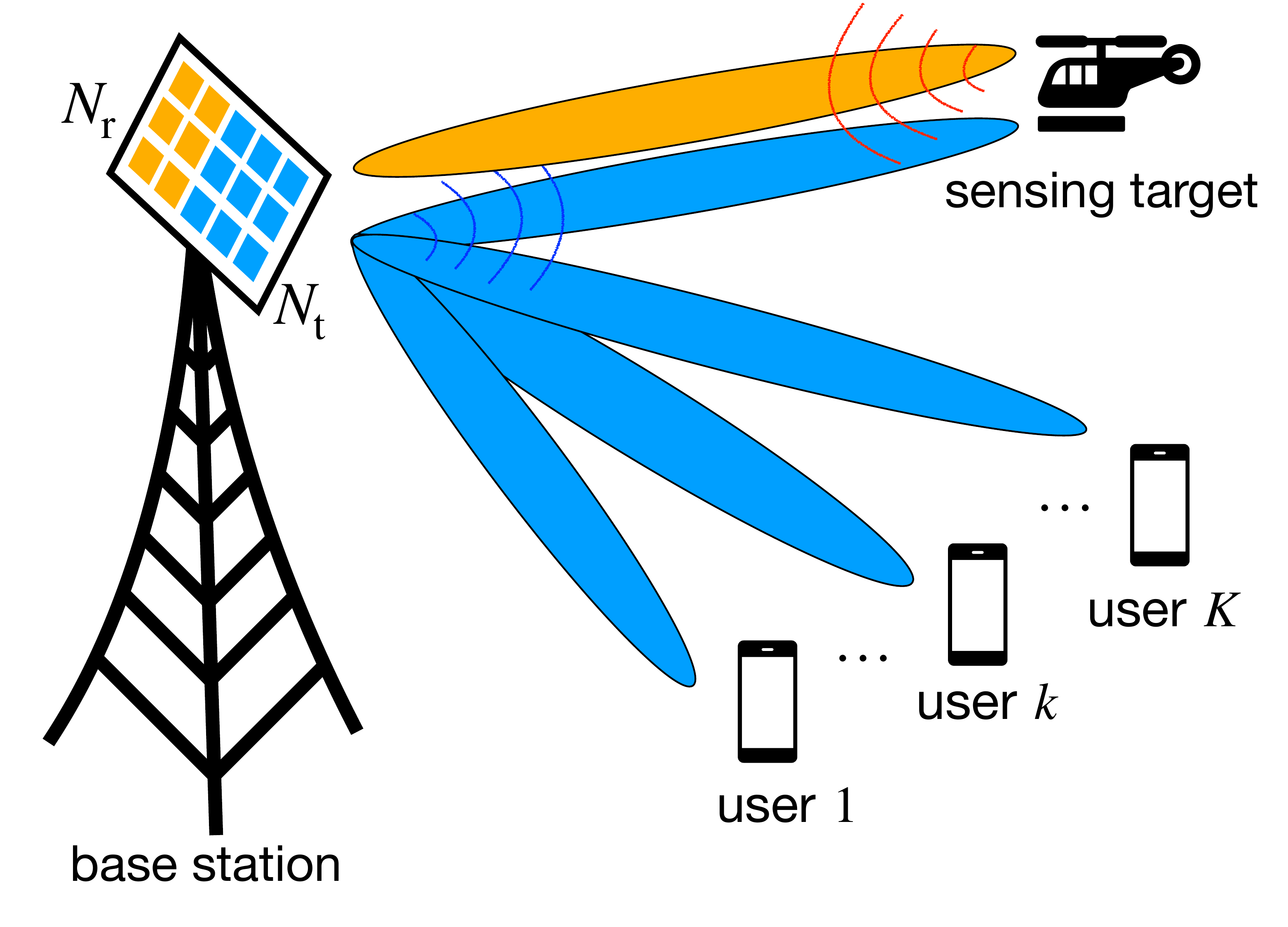}
		\caption{Mono-static massive MIMO ISAC system with a UPA at the BS.}
		\label{fig_system}
	\end{figure}
	
	Denote the transmit vector from the BS at the $\ell$-th time slot by $\vs_\ell = [s_{1\ell}, \ldots, s_{k\ell}, \ldots, s_{K\ell}] \in {\mathbb{C}}^{K \times 1}$, where $\mean{\vs_\ell \vs_\ell^{\H}}=\mI_K$ and $s_{k\ell}$ is the signal intended for the $k$-th user. Furthermore, let $\mS = [\vs_1, \ldots, \vs_L] \in \setC^{K \times L}$ be the transmit symbol matrix, where $L \gg 1$ is the length of the radar/communications frame. The data streams are assumed to be independent of each other such that
	\begin{align*}
		\mS \mS^\H \approx L \mI_K, \nbthis \label{eq_symbol_matrix}
	\end{align*}
	which becomes very tight for independent $\{\vs_1, \ldots, \vs_L\}$ and sufficiently large $L$ ~\cite{liu2021cramer} due to the law of large numbers. 
	
	The BS employs the linear precoder
	\begin{align*}
		\mF = \mW \bm{\Gamma} + \vv \bar{\bm{\eta}}^\T \in \mathbb{C}^{\Nt \times K}, \nbthis \label{eq_F}
	\end{align*}
	where $\mW = [\vw_1,\ldots,\vw_K] \in \mathbb{C}^{\Nt \times K}$ and $\bm{\Gamma} = \diag{\sqrt{\gamma_1}, \ldots, \sqrt{\gamma_K}} \in \setC^{K \times K}$ are the matrix of precoding vectors and power allocation factors, respectively, for the communications users. Specifically, $\vw_k$ and $\gamma_k$ are the precoding vector and power allocated for the $k$-th communications user. The vectors $\vv \in \mathbb{C}^{\Nt \times 1}$ and $\bar{\bm{\eta}} = [\sqrt{\eta_1},\ldots,\sqrt{\eta_K}]^\T \in \setC^{K \times 1}$ represent the precoding vector and power fraction allocated for sensing in each data stream, respectively. The $k$-th column of $\mF$, denoted as $\vf_k$, represents the dual-functional precoding vector for user $k$ and is given as $\vf_k = \sqrt{\gamma_k} \vw_k + \sqrt{\eta_k} \vv$.	Then, the $\Nt \times 1$ transmit signal vector during the $\ell$-th time slot is $\vx_\ell = \mF \vs_\ell = \sum_{k=1}^{K} \vf_k s_{k \ell}$, and the overall dual-functional transmit waveform is denoted by $\mX = [\vx_1, \ldots, \vx_L] \in \setC^{\Nt \times L}$. Equivalently, we have $\mX = \mF \mS$.
	
	For transmit waveform $\mX$, the $K \times L$ combined signal matrix received by the users can be expressed as
	\begin{align*}
		\mY_{\mathtt{c}} = \mH^\H \mX + \mN_{\mathtt{c}} = \mH^\H \mF \mS + \mN_{\mathtt{c}},\nbthis \label{eq_comm_model}
	\end{align*}
	where $\mN_{\mathtt{c}} \in \setC^{K \times L}$ is additive white Gaussian noise (AWGN) with independent entries following the distribution $\mathcal{CN}(0, \sigmac)$, {with $\sigmac$ denoting the noise variance at the communications user receivers}. Furthermore, $\mH = [\vh_1,\ldots,\vh_k,\ldots,\vh_K] \in \mathbb{C}^{\Nt \times K}$ is the channel matrix from the BS to the $K$ users. Here, $\vh_k$ denotes the channel between the BS and the $k$-th user, given as~\cite{mollen2016uplink}
	\begin{align}
		\vh_k = \beta_k^{1/2} \bar{\vh}_k, \label{eq_channel_model}
	\end{align} 
	where $\beta_k$ and $\bar{\vh}_k\sim \mathcal{CN}(0, \mI_{\Nt})$ represent the large-scale fading parameter and the small-scale Rayleigh fading channels, respectively. 
    At the $\ell$-th time slot, the received signal at user $k$ is
	\begin{align*}
		{y_{\mathtt{c}}}_{k\ell} = \vh_k^\H \vf_k s_{k\ell} + \vh_k^\H \sum\nolimits_{j \neq k} \vf_j s_{j \ell} + {n_{\mathtt{c}}}_{k\ell}. \nbthis \label{eq_y_kl}
	\end{align*}
	The linear precoding matrix $\mF$ is computed using the channel estimates acquired during the uplink training phase \cite{ngo2013massive}.  Let $\hat{\vh}_k$ and $\ve_k$ denote the channel estimate and the corresponding estimation error, such that $\vh_k = \hat{\vh}_k + \ve_k$. Assuming minimum mean square error (MMSE) estimation, $\hat{\vh}_k$ and $\ve_k$ are independent. In addition,  $\hat{\vh}_k \sim \mathcal{CN}(0, \xi_k \mI_{\Nt})$, and $\ve_k \sim \mathcal{CN}(0, \epsilon_k \mI_{\Nt})$, where~\cite{ngo2013massive}
	\begin{align*}
		\xi_k &= \frac{\tau_{\mathtt{p}} p_{\mathtt{p}} \beta_k^2}{\tau_{\mathtt{p}} p_{\mathtt{p}} \sum_{j=1}^{K} \beta_j  \abs{\vp_j^\H \vp_k}^2 + \sigmac}, \nbthis \label{eq_xi_k}\\
		\epsilon_k = \beta_k - \xi_k &= \frac{\beta_k (\tau_{\mathtt{p}} p_{\mathtt{p}} \sum_{j\neq k} \beta_j  \abs{\vp_j^\H \vp_k}^2 + \sigmac )}{\tau_{\mathtt{p}} p_{\mathtt{p}} \sum_{j=1}^{K} \beta_j  \abs{\vp_j^\H \vp_k}^2 + \sigmac}, \nbthis \label{eq_epsilon_k}
	\end{align*} 
        and where $\vp_k \in \setC^{\tau_{\mathtt{p}} \times 1}$, with $\|\vp_k\|^2=1$, is the pilot sequence transmitted by the $k$th user, $\tau_\mathtt{p}$ is the length of the pilot sequences, and $p_{\mathtt{p}}$ is the average power of the training symbols. The length of the training sequences must be shorter than the coherence interval $\tau_{\mathtt{c}}$.

	\subsection{Radar Model}
	While transmitting communications signals to the users, the BS also receives echo signals from the target. The received discrete-time radar sensing signal is given as \cite{liu2021cramer, song2023intelligent}
	\begin{align*}
		\mY_{\mathtt{s}} &= \alpha \mG(\theta, \phi) \mX + \mN_{\mathtt{s}} \in \setC^{\Nr \times L}, \nbthis \label{eq_radar_model}
	\end{align*}
	where $\alpha$ is the reflection coefficient, $\mN_{\mathtt{s}} \in \setC^{\Nr \times L}$ is an AWGN matrix with entries distributed as $\mathcal{CN}(0, \sigmas)$, {with $\sigmas$ denoting the noise variance at the sensing receiver}. Furthermore, $\mG(\theta, \phi) \in \setC^{\Nr \times \Nt}$ is the two-way channel in the desired sensing directions, modeled as~\cite{johnston2022mimo, liu2020joint}:
	\begin{align}
		\mG(\theta, \phi) &= \vb(\theta, \phi) \va^\H(\theta, \phi), \label{eq_A}
	\end{align} 
	where $\theta$ and $\phi$ are the azimuth and elevation angles of the target relative to the BS, {$\theta \in [-\pi,\pi]$, $\phi \in \left[-\frac{\pi}{2}, \frac{\pi}{2}\right]$}, and $\va(\theta, \phi)$ and $\vb(\theta, \phi)$ are the transmit and receive steering vectors at the BS, respectively. In the following analysis, we drop $(\theta, \phi)$ for ease of exposition. Let  $\Nth \times \Ntv$ be the size of the UPA, where $\Nth$ and $\Ntv$ are the numbers of antennas in the horizontal and vertical dimensions and $\Nth \Ntv = \Nt$. To facilitate the analysis, we choose the center of the UPA as the reference point~\cite{liu2021cramer, song2023intelligent, bekkerman2006target} and assume half-wavelength antenna spacing. As a result, $\va$ is modeled as~\cite{nguyen2022beam, yu2016alternating}
	\begin{align*}
		\va = \ah \otimes \av, \nbthis \label{eq_at}
	\end{align*} 
	where
	\begin{align*}
		&\ah = \Big[ e^{-j\pi\frac{\Nth-1}{2} \sin(\theta) \sin(\phi)}, e^{-j\pi\frac{\Nth-3}{2} \sin(\theta) \sin(\phi)}, \\
		&\qquad \ldots, e^{j\pi\frac{\Nth-3}{2} \sin(\theta) \sin(\phi)}, e^{j\pi\frac{\Nth-1}{2} \sin(\theta) \sin(\phi)} \Big]^\T, \nbthis \label{eq_ath} \\
		&\av = \Big[ e^{-j\pi\frac{\Ntv-1}{2} \cos(\phi)}, e^{-j\pi\frac{\Ntv-3}{2} \cos(\phi)}, \\
		&\hspace{2cm}  \ldots, e^{j\pi\frac{\Ntv-3}{2} \cos(\phi)}, e^{j\pi\frac{\Ntv-1}{2} \cos(\phi)} \Big]^\T, \nbthis \label{eq_atv}
	\end{align*}
	are the array response vectors corresponding to the horizontal and vertical dimensions of the UPA, respectively. The receive response vector $\vb$ for the $\Nr$ receive antennas at the BS is modeled similarly.

\subsection{Linear ISAC Beamforming}
	From \eqref{eq_F}, \eqref{eq_A}, and based on $\mX = \mF \mS$, we rewrite \eqref{eq_radar_model} as
	\begin{align*}
		\!\mY_{\mathtt{s}} = \alpha \vb \va^\H  \left(\mW \bm{\Gamma} + \vv \bar{\bm{\eta}}^\T\right) \mS + \mN_{\mathtt{s}} = \alpha \vb \va^\H \vv \bar{\bm{\eta}}^\T  \mS + \tilde{\mN}_{\mathtt{s}}, \nbthis \label{eq_radar_2}
	\end{align*}
	where $\tilde{\mN}_{\mathtt{s}} \triangleq \alpha \vb \va^\H \mW \bm{\Gamma} \mS + \mN_{\mathtt{s}}$. {For an arbitrary $\mW \bm{\Gamma}$, the optimal sensing beamformer would be $\vv = \va(\theta,\phi)$ if the target angles ($\theta, \phi$) were known. To maintain generality in the subsequent derivations, here we assume that $\vv=\av(\tilde{\theta},\tilde{\phi})$ for arbitrary $\tilde{\theta}$ and $\tilde{\phi}$, so that the performance can be evaluated for the more general case where these angles are unknown.}
    
	
	For communications, the MRT and ZF beamformers are considered, which are given by
	\begin{align*}
		\mW = 
		\begin{cases*}
			\hat{\mH} \triangleq \mW_{\mathtt{MRT}}, &\text{for MRT} \\
			\hat{\mH} (\hat{\mH}^\H \hat{\mH})^{-1} \triangleq \mW_{\mathtt{ZF}}  &\text{for ZF}.
		\end{cases*} \nbthis \label{eq_W_linear}
	\end{align*}
	The overall MRT and ZF beamforming matrices become
	\begin{align*}
		\mF_{\mathtt{bf}} &=  \mW_{\mathtt{bf}}  \bGamma + {\vv} \bar{\bm{\eta}}^\T, \nbthis \label{eq_F_MRT} 
	\end{align*}
	where $\mathtt{bf} \in \{\mathtt{MRT}, \mathtt{ZF}\}$. To facilitate the subsequent analysis and design, we first derive the total transmit power of the MRT and ZF precoders in the following lemma.
	
	\begin{lemma}
		\label{lemma_power}
		The total transmit power of the BS employing the MRT and ZF precoders is given as
		\begin{align*}
			P_{\mathtt{bf}} &= \Nt \xibf^\T \bm{\gamma} + \Nt \rho, \nbthis \label{eq_power}
		\end{align*}
		where  $\rho = \norm{\bar{\bm{\eta}}}^2$, and $\xibf$ is defined as
        \begin{align*}
				\xibf = 
				\begin{cases} \begin{array}{ll}
					\!\!\![\xi_1,\ldots,\xi_K]^\T \triangleq \bm{\xi}_{\mathtt{MRT}} &\!\!\! \text{for~MRT} \\
					\!\!\!\frac{1}{\Nt(\Nt-K)}\left[\frac{1}{\xi_1},\ldots,\frac{1}{\xi_K}\right]^\T \triangleq \bm{\xi}_{\mathtt{ZF}} & \!\!\!\text{for~ZF}, \end{array}
				\end{cases} \nbthis \label{eq_def_xi_bf}
			\end{align*}
   with $\mathtt{bf} \in \{\mathtt{MRT}, \mathtt{ZF}\}$.
	\end{lemma}
        \begin{proof}
		See Appendix \ref{appd_power}.
	\end{proof}
	
	\section{Communications and Sensing Performance with MRT and ZF Precoders}
	\label{sec_perf_analysis}
	In this section, we derive lower bounds on the achievable rate of the communications subsystem with the MRT and ZF precoders. For the sensing subsystem, we derive the CRLB to evaluate the estimation accuracy for $\theta$ and $\phi$.

\subsection{Communications Performance}
	We rewrite \eqref{eq_y_kl} as
	\begin{align*}
		{y_{\mathtt{c}}}_{k\ell} &= \mathtt{DS}_k s_{k\ell} + \mathtt{BU}_k s_{k\ell} + \sum\nolimits_{j \neq k} \mathtt{UI}_{kj} s_{j\ell} + {n_{\mathtt{c}}}_{k\ell}, \nbthis \label{eq_ykl_MRT}
	\end{align*}
	where $\mathtt{DS}_k \triangleq \mean{ \vh_k^\H \vf_k}$, $ \mathtt{BU}_k \triangleq \vh_k^\H \vf_k - \mean{ \vh_k^\H \vf_k}$, and $\mathtt{UI}_{kj} \triangleq \vh_k^\H  \vf_j$ represent expectations associated with the desired signal, beamforming uncertainty, and inter-user interference, respectively. From \eqref{eq_ykl_MRT}, the achievable rate of the $k$-th user is given by
	\begin{align*}
		R_{k} = \bar{\tau} \log_2\!\left(1 + \frac{\abs{\mathtt{DS}_k}^2}{\mean{\abs{\mathtt{BU}_k}^2} +  \sum\limits_{j \neq k} \mean{\abs{\mathtt{UI}_{kj}}^2} + \sigmac} \right), \nbthis \label{eq_SE_def}
	\end{align*}
	where $\bar{\tau} \triangleq \left(\tau_{\mathtt{c}} - \tau_{\mathtt{p}}\right)/\tau_{\mathtt{c}}$. The following theorem presents the closed form expression for the achievable rate of the communications users.
	
	\begin{theorem}
		\label{theo_SE}
		The achievable rate for the $k$-th user with MRT or ZF precoding is given by
		\begin{align*}
			{R_{\mathtt{bf}}}_k(\bm{\gamma},\rho) &= \bar{\tau} \log_2 \left(1 + \frac{ \lambdabf_k \gamma_k }{\Nt \beta_k \rho  + \Nt {\bm{\zeta}_{\mathtt{bf}}^\T}_k \bm{\gamma} + \sigmac}\right), \nbthis \label{eq_SE_theo}
		\end{align*}
		where $\bm{\gamma} = [\gamma_1,\ldots,\gamma_K]^\T$, $\mathtt{bf} \in \{\mathtt{MRT}, \mathtt{ZF}\}$, and
		\begin{align*}
			\lambdabf_k &= 
			\begin{cases}
				\begin{array}{ll} \Nt^2 \xi_k^2 \triangleq {\lambda_{\mathtt{MRT}}}_k &  \text{for~MRT} \\
				1 \triangleq {\lambda_{\mathtt{ZF}}}_k & \text{for~ZF} \end{array}
			\end{cases} \nbthis \label{eq_def_lambda_bf} \\
			\zetabf_k \!&=\! 
			\begin{cases}
				\begin{array}{ll} 
                    \!\!\!\!\!\beta_k \left[\xi_1, \ldots, \xi_K\right]^\T \!\triangleq\! {\bm{\zeta}_{\mathtt{MRT}}}_k & \!\!\!\!\text{for~MRT} \\
				\!\!\!\!\!\frac{\epsilon_k}{\Nt(\Nt-K)} \left[ \frac{1}{\xi_1}, \ldots, \frac{1}{\xi_K} \right]^\T \!\triangleq \!{\bm{\zeta}_{\mathtt{ZF}}}_k & \!\!\!\!\text{for~ZF}. \end{array}
			\end{cases} \nbthis \label{eq_def_zeta_bf}
		\end{align*}
	\end{theorem}
	
	\begin{proof}
		See Appendix \ref{appd_SE}. \epr
	\end{proof}
        \begin{remark}
		\label{rm_effect_of_sensing}
		It is observed from \eqref{eq_SE_theo} that, due to the term $\Nt \beta_k \rho$ in the denominator of the SINR, sensing increases beamforming uncertainty and inter-user interference, and thus causes performance degradation to the communications subsystem for both MRT and ZF. 
	\end{remark}
 
        To obtain basic insights into the impact of an increasing number of antennas on the communications performance, we consider a simplified case where the transmit power is equally shared between communications and sensing, and the communications users all have the same power control factors, i.e., $\gamma_1 = \ldots = \gamma_K$. In this case, based on \eqref{eq_power}, we have
        \begin{align*}
            \rho = \frac{P_{\mathtt{bf}}}{2\Nt},\ \gamma_k = \frac{P_{\mathtt{bf}}}{2\Nt \sum_{j=1}^K {\xi_{\mathtt{bf}}}_j},\ \forall k, \nbthis \label{eq_equal_power_solution}
        \end{align*}
        where ${\xi_{\mathtt{bf}}}_k$ is the $k$-th element of $\xibf$. Accordingly, the achievable rate of the $k$-th user with MRT and ZF precoding is given as
        \begin{align*}
            \bar{R}_{\mathtt{MRT}k} &= \bar{\tau} \log_2 \left(1 + \frac{ \Nt \xi_k^2 P_{\mathtt{bf}} }{ 2 \left( \beta_k P_{\mathtt{bf}} + \sigmac \right) \sum_{j=1}^K \xi_k }\right),\ \nbthis \label{eq_R_MRT} \\
            \bar{R}_{\mathtt{ZF}k} &= \bar{\tau} \log_2 \left(1 + \frac{ (\Nt-K)P_{\mathtt{bf}} }{\left( \left(\beta_k + \epsilon_k\right) P_{\mathtt{bf}} + 2\sigmac \right) \sum_{j=1}^K \frac{1}{\xi_j} }\right).\ \nbthis \label{eq_R_ZF}
        \end{align*}
        \begin{remark}
            \label{rm_free_interference}
            In \eqref{eq_R_MRT} and \eqref{eq_R_ZF}, the numerators of the SINR terms increase with $\Nt$, while the denominators do not. Therefore, massive MIMO ISAC systems with both MRT and ZF precoding and equal power allocation become interference-free as $\Nt \rightarrow \infty$, similar to conventional massive MIMO systems without any sensing function. In other words, deploying a very large number of transmit antennas can mitigate the impact of sensing on communications performance.
        \end{remark}

    Although Remark \ref{rm_free_interference} is obtained for the equal power allocation in \eqref{eq_equal_power_solution}, we will numerically show in Section \ref{sec_sim} that it is also valid for other considered power allocation schemes. Next, we investigate the sensing performance. 
	
	\subsection{Sensing Performance}
	
	In the sensing subsystem, we are interested in characterizing the CRLB for estimating the target angles, i.e., $(\theta, \phi)$. Denote $\vy_{\mathtt{s}} = \mathtt{vec}(\mY_{\mathtt{s}}) \in \setC^{\Nr L \times 1}$, $\vx_{\mathtt{s}} \triangleq \mathtt{vec}(\mG \mX) \in \setC^{\Nr L \times 1}$, $\vn_{\mathtt{s}} = \mathtt{vec}(\mN_{\mathtt{s}}) \in \setC^{\Nr L \times 1}$, and $\vn_{\mathtt{s}} \sim \mathcal{CN}(\bm{0}, \sigmas \mI_{\Nr L})$, and rewrite \eqref{eq_radar_model} as
	\begin{align*}
		\vy_{\mathtt{s}} = \alpha \vx_{\mathtt{s}} + \vn_{\mathtt{s}}. \nbthis \label{eq_radar_model_2}
	\end{align*}
	Let $\bm{\omega} \triangleq [\theta, \phi, \bar{\bm{\alpha}}]^\T \in \setR^{4 \times 1}$ be the vector of unknown parameters to be estimated from \eqref{eq_radar_model_2}, where $\bar{\bm{\alpha}} = [\re{\alpha}, \im{\alpha}]^\T$. The maximum likelihood estimate (MLE) of $\bm{\omega}$ is given as
	\begin{align*}
		\bm{\omega}_{\mathtt{MLE}} = \argmin_{\bm{\omega}} \norm{\vy_{\mathtt{s}} - \alpha \vx_{\mathtt{s}}}^2. \nbthis \label{eq_MLE}
	\end{align*}
The Fisher information matrix for estimating $\bm{\omega}$, denoted by $\mJ_{\bm{\omega}}$, can be obtained as~\cite{liu2021cramer, song2023intelligent, bekkerman2006target}
	\begin{align*}
		\mJ_{\bm{\omega}} = \re{\frac{\partial \vx_{\mathtt{s}}^\H}{\partial \bm{\omega}} \mR_{\mathtt{s}}^{-1} \frac{\partial \vx_{\mathtt{s}}}{\partial \bm{\omega}}} = \frac{2}{\sigmas} \re{\frac{\partial \vx_{\mathtt{s}}^\H}{\partial \bm{\omega}}  \frac{\partial \vx_{\mathtt{s}}}{\partial \bm{\omega}}}, \nbthis \label{eq_FIMa}
	\end{align*}
	where $\mR_{\mathtt{s}} = \sigmas \mI_{\Nr L}$. We can partition $\mJ_{\bm{\omega}}$ into blocks as
	\begin{align*}
		\mJ_{\bm{\omega}} = 
		\begin{bmatrix}
			\Jtt  & \Jtp & \Jta\\
			\Jtp & \Jpp & \Jpa \\
			\Jta^\T & \Jpa^\T &\Jaa
		\end{bmatrix},
	\end{align*}
	with the (block) elements of $\mJ_{\bm{\omega}}$ obtained as~\cite{bekkerman2006target}
	\begin{align*}
		\Jtt &= \kappa \abs{\alpha}^2 \tr{\dot{\mG}_{\theta} \mR_x \dot{\mG}^\H_{\theta}} \in \setR^{1 \times 1}, \nbthis \label{eq_def_Jtt} \\
		\Jtp &= \kappa \abs{\alpha}^2 \tr{\dot{\mG}_{\phi} \mR_x \dot{\mG}^\H_{\theta}} \in \setR^{1 \times 1}, \nbthis \label{eq_def_Jtp}\\
		\Jta &= \kappa \re{\alpha^* \tr{\mG \mR_x \dot{\mG}^\H_{\theta}}[1,j]} \in \setR^{1 \times 2}, \nbthis \label{eq_def_Jta}\\
		\Jpp &= \kappa \abs{\alpha}^2 \tr{\dot{\mG}_{\phi} \mR_x \dot{\mG}^\H_{\phi}} \in \setR^{1 \times 1}, \nbthis \label{eq_def_Jpp} \\
		\Jpa &= \kappa \re{\alpha^* \tr{\mG \mR_x \dot{\mG}^\H_{\phi}}[1,j]} \in \setR^{1 \times 2}, \nbthis \label{eq_def_Jpa}\\
		\Jaa &= \kappa \tr{\mG \mR_x \mG^\H} \mI_2  \in \setR^{2 \times 2}, \nbthis \label{eq_def_Jaa}
	\end{align*}
where
	\begin{align*}
		&\dot{\mG}_{\theta} = \bdottheta \va^\H + \vb \adottheta^\H,\ 
		\dot{\mG}_{\phi} = \bdotphi \va^\H + \vb \adotphi^\H, \nbthis \label{eq_dot_G} \\
            &\mR_x = \frac{1}{L} \mean{\mX \mX^\H}  = \!\frac{1}{L} \mean{\mF_{\mathtt{bf}} \mS \mS^\H \mF_{\mathtt{bf}}^\H}\\  
            &\hspace{3cm}= \mean{\mF_{\mathtt{bf}} \mF_{\mathtt{bf}}^\H} \triangleq \mR_{\mathtt{bf}}, \nbthis\label{eq_approx_cov}
	\end{align*}
	and $\kappa \triangleq \frac{2L}{\sigmas}$. The equivalent FIM for $(\theta, \phi)$ is \cite{huang2022joint}
	\begin{align*}
		\mJ_{\theta, \phi} = 
		\begin{bmatrix}
			\Jtt  & \Jtp \\
			\Jtp & \Jpp - \Jpa \Jaa^{-1} \Jpa^\T
		\end{bmatrix} = 
		\begin{bmatrix}
			\Jtt  & \Jtp \\
			\Jtp & \Jpat
		\end{bmatrix}, \nbthis \label{eq_FIM}
	\end{align*}
	where
    \begin{align*}
        \Jpat \triangleq \Jpp - \Jpa \Jaa^{-1} \Jpa^\T. \nbthis \label{eq_Jpat}
    \end{align*}
    The diagonal entries of $\mJ_{\theta, \phi}^{-1}$ serve as lower bounds on the MSE of any unbiased estimate of $\theta$ and $\phi$. As a result, we can derive the CRLBs associated with $\theta$ and $\phi$ with MRT and ZF precoding in the following theorem.
    
\begin{theorem}
    \label{theo_CRB}
    The CRLBs for $\theta$ and $\phi$ in the considered massive MIMO ISAC system for either MRT or ZF precoding admits the common closed form expressions as follows: 
    \begin{align*}
		&\widetilde{\mathtt{CRLB}}_{\mathtt{bf},\theta}(\bm{\gamma},\rho) = \left(\Jtt - \frac{\Jtp^2}{\Jpat}\right)^{-1}, \nbthis \label{eq_CRB_theta_0}\\
		&\widetilde{\mathtt{CRLB}}_{\mathtt{bf},\phi}(\bm{\gamma},\rho)  = \left(\Jpat - \frac{\Jtp^2}{\Jtt}\right)^{-1}, \nbthis \label{eq_CRB_phi_0}
	\end{align*}
    where
	    \begin{align*}
        \Jtt &= \kappa \abs{\alpha}^2 \Big(\bm{\xi}_{\mathtt{bf}}^\T \bm{\gamma} \left( \Nr \norm{\adottheta}^2 + \Nt \normshort{\bdottheta}^2\right)\\
        &\hspace{1.5cm} + \rho {\left( {\norm{\vv^\H \va}^2} \normshort{\bdottheta}^2 + \Nr {\norm{\vv^\H \adottheta}^2} \right)} \Big). \nbthis \label{eq_Jtt} \\
		\Jpp\! &=\! \kappa\! \abs{\alpha}^2 \Big( \bm{\xi}_{\mathtt{bf}}^\T \bm{\gamma} \left(\Nr \norm{\adotphi}^2\! +\! \Nt \normshort{\bdotphi}^2\right) \\
        &\hspace{1.5cm} + \rho {\left( {\norm{\vv^\H \va}^2} \normshort{\bdotphi}^2 + \Nr {\norm{\vv^\H \adotphi}^2}  \right)}\Big), \nbthis \label{eq_Jpp}  \\
		\Jtp\! &= \kappa \abs{\alpha}^2 \Big(\bm{\xi}_{\mathtt{bf}}^\T \bm{\gamma} \left( \Nr \adottheta^\H \adotphi  + \Nt \bdottheta^\H \bdotphi\right) + \\
        &\hspace{1.5cm} + \rho {\left( {\norm{\vv^\H \va}^2} \bdotphi^\H \bdottheta + \Nr  \adottheta^\H \vv \vv^\H \adotphi \right)}\Big), \nbthis \label{eq_Jtp}  \\
		\Jaa &= \kappa \left( \bm{\xi}_{\mathtt{bf}}^\T \bm{\gamma} \Nt \Nr + \rho \Nr {\norm{\vv^\H \va}^2}  \right)  \mI_2, \nbthis \label{eq_Jaa} \\
        \Jpa &= \kappa \rho \Nr \re{\va^\H {\vv \vv^\H} \adotphi [1,j]}, \nbthis \label{eq_Jpa}
	\end{align*}
    and $\Jpat$ is computed based on \eqref{eq_Jpat}. Here, $\mathtt{bf} \in \{\mathtt{MRT}, \mathtt{ZF}\}$, and $\xibf$ is defined in \eqref{eq_def_xi_bf}. 
	\end{theorem}

    \begin{proof}
		See Appendix \ref{appd_CRB}. \epr
	\end{proof}
    \begin{figure}[t]
        \centering
        \includegraphics[scale=0.55]{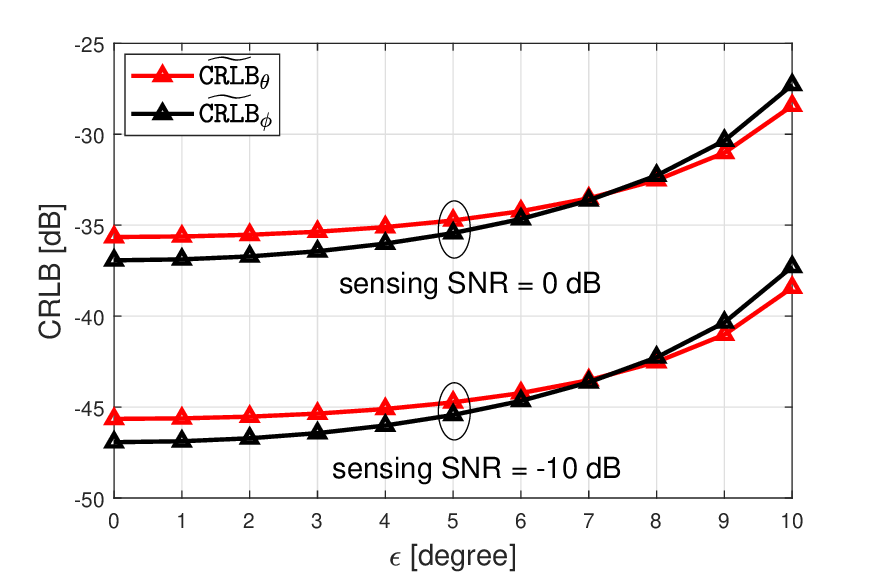}
        \caption{{CRLB for $\theta$ and $\phi$ when the sensing beamformer $\vv$ is determined by $\tilde{\theta} = \theta + \epsilon$ and $\tilde{\phi} = \phi + \epsilon$, $\Nt=25$, $\Nr=25$, $K=8$, $L=30$.}}
        \label{fig_CRB_general}
    \end{figure}

    As explained earlier, the CRLB in \eqref{eq_CRB_theta_0} and \eqref{eq_CRB_phi_0} will depend on the choice of the sensing beamformer $\vv$. Fig.\ \ref{fig_CRB_general} shows the CRLB obtained from Theorem \ref{theo_CRB} when $\Nt= \Nr=25$, $K=8$, $L=30$, and $\vv$ is chosen as $\vv = \va(\tilde{\theta},\tilde{\phi})$, where $\tilde{\theta} = \theta + \epsilon$, $\tilde{\phi} = \phi + \epsilon$, and $\epsilon \in [0^\circ, 10^\circ]$. 
    As expected, the target angle estimation performance deteriorates when the sensing beamformer is not pointed directly at the target, but we see that the CRLB is not sensitive to misalignment of the sensing beam.
    
     The CRLB expressions in \eqref{eq_CRB_theta_0} and \eqref{eq_CRB_phi_0} are complicated due to the intricate structure of $\Jtt$, $\Jpp$, $\Jtp$, $\Jaa$, $\Jpa$, and especially $\Jpat$, making them intractable for analysis and even more so for optimization-based system design. In the following remark, we simplify the CRLB expression for the case where $(\tilde{\theta}, \tilde{\phi}) = (\theta, \phi)$, i.e., when the sensing beam is perfectly aligned with the target angles. In this case, we have $\vv = \va$. While this assumption does not affect the closed form expression for the achievable rate in Theorem \ref{theo_SE}, it significantly simplifies the CRLB expressions, facilitating both analysis and design.
     
	\begin{figure*}
		\begin{align*}
			&\mathtt{CRLB}_{\mathtt{bf},\theta}(\bm{\gamma},\rho) = \frac{1}{\kappa \abs{\alpha}^2 } \left(\xibf^\T \bm{\gamma} \left( \Nr \norm{\adottheta}^2 + \Nt \normshort{\bdottheta}^2\right) + \rho \Nt^2 \normshort{\bdottheta}^2 - \frac{\left(\xibf^\T \bm{\gamma} \left(\Nr \adottheta^\H \adotphi  + \Nt \bdottheta^\H \bdotphi\right) + \rho \Nt^2 \bdottheta^\H \bdotphi\right)^2}{ \xibf^\T \bm{\gamma}\left( \Nr \norm{\adotphi}^2 + \Nt \normshort{\bdotphi}^2\right) + \rho \Nt^2 \normshort{\bdotphi}^2}\right)^{-1} \nbthis \label{eq_CRB_theta_1}\\
			&\mathtt{CRLB}_{\mathtt{bf},\phi}(\bm{\gamma},\rho) = \frac{1}{\kappa \abs{\alpha}^2 }  \left(\xibf^\T \bm{\gamma}\left( \Nr \norm{\adotphi}^2 + \Nt \normshort{\bdotphi}^2\right) + \rho \Nt^2 \normshort{\bdotphi}^2 - \frac{\left(\xibf^\T \bm{\gamma} \left(\Nr \adottheta^\H \adotphi  + \Nt \bdottheta^\H \bdotphi\right) + \rho \Nt^2 \bdottheta^\H \bdotphi\right)^2}{\xibf^\T \bm{\gamma} \left( \Nr \norm{\adottheta}^2 + \Nt \normshort{\bdottheta}^2\right) + \rho \Nt^2 \normshort{\bdottheta}^2}\right)^{-1} \nbthis \label{eq_CRB_phi_1}
		\end{align*}
		\hrule
	\end{figure*}

	\begin{remark}
			\label{rm_CRB_0}
			{For the case where $\vv = \va$}, the CRLB for $\theta$ and $\phi$ in the considered massive MIMO ISAC system for either MRT or ZF precoding can be simplified to \eqref{eq_CRB_theta_1} and \eqref{eq_CRB_phi_1}, where 
            \begin{align*}
				\norm{\adottheta}^2 &= \frac{\Nt(\Nth^2-1)}{12}  \pi^2 \cos^2(\theta) \sin^2(\phi) , \nbthis \label{eq_norm_adot_theta} \\
				\norm{\adotphi}^2 	&= \frac{\Nt}{12} \pi^2 \cos^2(\phi)  \left((\Nth^2-1) \sin^2(\theta)  + (\Ntv^2-1)\right), \nbthis \label{eq_norm_adot_phi}\\
				\normshort{\bdottheta}^2 &= \frac{\Nr(\Nrh^2-1)}{12} \pi^2 \cos^2(\theta) \sin^2(\phi) , \nbthis \label{eq_norm_bdot_theta} \\
				\normshort{\bdotphi}^2 	&= \frac{\Nr}{12} \pi^2 \cos^2(\phi)  \left((\Nrh^2-1) \sin^2(\theta)  + (\Nrv^2-1)\right), \nbthis \label{eq_norm_bdot_phi}\\
				\adottheta^\H \adotphi &=  \frac{\Nt(\Nth^2-1)}{12} \pi^2 \sin(\phi) \sin(\theta) \cos(\phi) \cos(\theta) , \nbthis \label{eq_adot_theta_adot_phi}\\
				\bdottheta^\H \bdotphi &= \frac{\Nr(\Nrh^2-1)}{12} \pi^2 \sin(\phi) \sin(\theta) \cos(\phi) \cos(\theta). \nbthis \label{eq_bdot_theta_bdot_phi}
            \end{align*}
	\end{remark}
	
	\begin{proof}
		See Appendix \ref{appd_proof_of_remark_CRB}. \epr
	\end{proof}

        {Based on the closed form expressions in Remark \ref{rm_CRB_0},} we characterize important properties of the CRLB in the following remarks.
        \begin{remark}
		\label{rm_CRB_2}
		Different power allocation schemes can achieve the same CRLB as long as they provide the same power fraction between communications and sensing, i.e., the same $\bm{\xi}^\T \bm{\gamma}$ and $\rho$. The power allocation among communications users, determined by $\{\gamma_k\}_{k=1}^K$, does not affect the CRLB.
	\end{remark}
 

        \begin{remark}
            \label{rm_CRB_1}
            Consider a square UPA, i.e., $\Nth = \Ntv = \sqrt{\Nt} \in \mathbb{N}$, and consider the equal power allocation in \eqref{eq_equal_power_solution}. In this case, the CRLB for both $\theta$ and $\phi$ decrease to zero when $\Nt \rightarrow \infty$.
        \end{remark}

        \begin{proof}
            See Appendix \ref{appd_CRB_go_to_zero}. \epr
        \end{proof}

    We have shown in Remarks \ref{rm_free_interference} and \ref{rm_CRB_2} that with the simplified equal power allocation in \eqref{eq_equal_power_solution}, the interference is mitigated and a very low CRLB can be achieved by a very large $\Nt$. However, this does not guarantee a good communications--sensing performance tradeoff. In particular, Remark \ref{rm_effect_of_sensing} implies that the massive MIMO ISAC system requires a proper power allocation between communications and sensing; otherwise, there will be significant performance loss due to sensing. In the next section, we present the power allocation problem and its solution.
    	
	\section{Communications and Sensing Power Allocation} 
	\label{sec_opt}
        \subsection{Problem Formulation}
	We are interested in a communications-centric design to maximize the communications rate while ensuring constraints on the sensing CRLB and the total transmit power. It is worth noting from \eqref{eq_SE_theo}, \eqref{eq_CRB_theta_1}, and \eqref{eq_CRB_phi_1} that both the achievable rate and CRLB can be determined using only $\rho$ rather than the sensing power factors $\{\eta_1,\ldots,\eta_K\}$. Here we recall that $\rho = \norm{\bar{\bm{\eta}}}^2 = \sum_{k=1}^K \eta_k$. With this observation, the optimization of $\{\gamma_k, \eta_k\}_{k=1}^K$ reduces to optimizing $\{\gamma_k\}_{k=1}^K$ and $\rho$. Using the result in Lemma \ref{lemma_power}, the power allocation problems for MRT and ZF precoding can be formulated as:
	\begin{subequations}\label{ori_problem_MRTZF}
		\begin{IEEEeqnarray} {rcl}
			(\mathcal{P}_{\mathtt{bf}}): \quad & \underset{\substack{ \bm{\gamma}, \rho }}{\textrm{maximize}}  & \quad  \sum_{k=1}^K {R_{\mathtt{bf}}}_k(\bm{\gamma}, \rho)  \label{eq::probMRT::obj_func_MRT} \\
			&\textrm{subject to} 
			&\quad \mathtt{CRLB}_{\mathtt{bf},\theta}(\bm{\gamma},\rho) \leq \mathtt{CRLB}_{\theta}^0  \label{eq::probMRT::cons_CRB_theta} \\
			& &\quad\mathtt{CRLB}_{\mathtt{bf},\phi}(\bm{\gamma},\rho) \leq \mathtt{CRLB}_{\phi}^0  \label{eq::probMRT::cons_CRB_phi} \\
			& &\quad\Nt \left( \xibf^\T \bm{\gamma} + \rho \right) \leq P_{\mathtt{t}}. \label{eq::probMRT::cons_power_MRT}
		\end{IEEEeqnarray}
	\end{subequations}
    {Note that the proposed optimization in \eqref{ori_problem_MRTZF} is performed based on the assumption of Remark~\ref{rm_CRB_0} that the CRLB constraints in \eqref{eq::probMRT::cons_CRB_theta} and \eqref{eq::probMRT::cons_CRB_phi} are computed without any pointing error. Unlike the sensing performance, there should be little impact on the communication sum-rate objective of \eqref{eq::probMRT::obj_func_MRT} when a small pointing error exists, since it is the presence of the CRLB constraint and not the precise angles associated with the constraint that has the primary effect.}
	
	\subsection{Proposed Solution to \eqref{ori_problem_MRTZF}}
	The objective \eqref{eq::probMRT::obj_func_MRT} is a nonconcave function in $(\bm{\gamma},\rho)$, while constraints \eqref{eq::probMRT::cons_CRB_theta}, \eqref{eq::probMRT::cons_CRB_phi} and \eqref{eq::probMRT::cons_power_MRT}	are convex and linear, so problem \eqref{ori_problem_MRTZF} is non-convex. To tackle this problem, we apply SCA as described below.
	\subsubsection{Convexifying Objective Function \eqref{eq::probMRT::obj_func_MRT}} We begin with the convex approximation of \eqref{eq::probMRT::obj_func_MRT} using the inner approximation (IA) framework~\cite{Beck:JGO:10}. It is clear that both the numerator and denominator of \eqref{eq::probMRT::obj_func_MRT} are linear in $(\bm{\gamma},\rho)$. Thus, we adopt the following inequality~\cite{NasirTWC21}:
	\begin{align}\label{obj_func_BB_appro}
		\ln\Big(1+\frac{x}{y}\Big) \geq & \ln\Big(1+\frac{x^{(i)}}{y^{(i)}}\Big) + 2 \frac{x^{(i)}}{x^{(i)}+y^{(i)}} \nonumber\\
		& - \frac{(x^{(i)})^2}{x^{(i)}+y^{(i)}}\frac{1}{x} - \frac{x^{(i)}}{(x^{(i)}+y^{(i)})y^{(i)}}y,
	\end{align}
	where $x^{(i)}$ and $y^{(i)}$ are respectively feasible points for $x$ and $y$ at the $i$-th iteration of an algorithm presented shortly. Given ($x^{(i)}$, $y^{(i)}$), we can see that the right-hand side (RHS) of \eqref{obj_func_BB_appro} is a concave lower bound for $\ln\big(1+x/y\big)$. Therefore, with $\lambdabf_k$ in \eqref{eq_def_lambda_bf} and $\zetabf_k$ in \eqref{eq_def_zeta_bf}, ${R_{\mathtt{bf}}}_k(\bm{\gamma},\rho)$ in \eqref{eq_SE_theo} is lower bounded around the point ($\bm{\gamma}^{(i)}$, $\rho^{(i)}$) at iteration $i$ as
	\begin{align*}\label{eq_SE_theo_concave}
		{R_{\mathtt{bf}}^{(i)}}_k(\bm{\gamma},\rho) = & \frac{\bar{\tau}}{\ln2}\left[ {A_{\mathtt{bf}}^{(i)}}_k - \frac{{B_{\mathtt{bf}}^{(i)}}_k}{\lambdabf_k}\frac{1}{\gamma_k} \right. \nonumber\\
		&\quad \left. - {C_{\mathtt{bf}}^{(i)}}_k\Bigl(\Nt \beta_k \rho + \Nt {\bm{\zeta}_{\mathtt{bf}}^\T}_k \bm{\gamma} + \sigmac\Bigl)\right],
		\nbthis 
	\end{align*}
	where
	\begin{align}
		{A_{\mathtt{bf}}^{(i)}}_k &\triangleq \ln\left(1 + \frac{ \lambdabf_k \gamma_k^{(i)} }{\Nt \beta_k \rho^{(i)}  + \Nt {\bm{\zeta}_{\mathtt{bf}}^\T}_k \bm{\gamma}^{(i)} + \sigmac}\right) \nonumber\\
		&\quad + 2 \frac{\lambdabf_k \gamma_k^{(i)}}{\lambdabf_k \gamma_k^{(i)} + \Nt \beta_k \rho^{(i)}  + \Nt {\bm{\zeta}_{\mathtt{bf}}^\T}_k \bm{\gamma}^{(i)} + \sigmac},\nonumber\\
		{B_{\mathtt{bf}}^{(i)}}_k &\triangleq \frac{\big(\lambdabf_k \gamma_k^{(i)}\big)^2 }{\lambdabf_k \gamma_k^{(i)} + \Nt \beta_k \rho^{(i)}  + \Nt {\bm{\zeta}_{\mathtt{bf}}^\T}_k \bm{\gamma}^{(i)} + \sigmac},\nonumber\\
		{C_{\mathtt{bf}}^{(i)}}_k &\triangleq  \lambdabf_k \gamma_k^{(i)}\Big/\Big[\Big(\lambdabf_k \gamma_k^{(i)} + \Nt \beta_k \rho^{(i)}  + \Nt {\bm{\zeta}_{\mathtt{bf}}^\T}_k \bm{\gamma}^{(i)} \nonumber\\
		&\qquad + \sigmac\Big)\Big(\Nt \beta_k \rho^{(i)}+\Nt {\bm{\zeta}_{\mathtt{bf}}^\T}_k \bm{\gamma}^{(i)} + \sigmac\Big)\Big].
	\end{align}
We note that ${R_{\mathtt{bf}}}_k(\bm{\gamma}^{(i)},\rho^{(i)})={R_{\mathtt{bf}}^{(i)}}_k(\bm{\gamma}^{(i)},\rho^{(i)})$.
	
	\subsubsection{SOC Transformation of \eqref{eq::probMRT::cons_CRB_theta} and \eqref{eq::probMRT::cons_CRB_phi}} Towards an efficient optimization method, we will convert \eqref{eq::probMRT::cons_CRB_theta} and \eqref{eq::probMRT::cons_CRB_phi} into second-order cone (SOC) constraints. 
	First, we can rewrite \eqref{eq::probMRT::cons_CRB_theta} equivalently as
	\begin{align}\label{eq_cons_CRB_theta_eqv}
		&\frac{\Big(\xibf^\T \bm{\gamma} \big(\Nr \adottheta^\H \adotphi  + \Nt \bdottheta^\H \bdotphi\big) + \rho \Nt^2 \bdottheta^\H \bdotphi\Big)^2}{ \xibf^\T \bm{\gamma}\left( \Nr \norm{\adotphi}^2 + \Nt \normshort{\bdotphi}^2\right) + \rho \Nt^2 \normshort{\bdotphi}^2}\nonumber\\
		&\leq \xibf^\T \bm{\gamma} \Big( \Nr \norm{\adottheta}^2\! +\! \Nt \normshort{\bdottheta}^2\Big)\! +\! \rho \Nt^2 \normshort{\bdottheta}^2 \!- \! \frac{1}{\kappa \abs{\alpha}^2\mathtt{CRLB}_{\theta}^0}.
	\end{align}
	Define $\varphi_{\theta} \triangleq \xibf^\T \bm{\gamma} \Big( \Nr \norm{\adottheta}^2 + \Nt \normshort{\bdottheta}^2\Big) + \rho \Nt^2 \normshort{\bdottheta}^2$ and $\varpi_{\phi} \triangleq  \xibf^\T \bm{\gamma}\left( \Nr \norm{\adotphi}^2 + \Nt \normshort{\bdotphi}^2\right) + \rho \Nt^2 \normshort{\bdotphi}^2$, so that the SOC formulation of \eqref{eq_cons_CRB_theta_eqv} can be written as
	\begin{align}\label{eq_cons_CRB_theta_SOC}
		&\Big\|\Big[\xibf^\T \bm{\gamma} \big(\Nr \adottheta^\H \adotphi  + \Nt \bdottheta^\H \bdotphi\big) + \rho \Nt^2 \bdottheta^\H \bdotphi;\; 0.5\Big(\varpi_{\phi}-\varphi_{\theta}\nonumber\\
		&\!+\!\frac{1}{\kappa \abs{\alpha}^2\mathtt{CRLB}_{\theta}^0}\Big)\!\Big]\!\Big\|_2\! \leq \!0.5\Big(\varpi_{\phi}+\varphi_{\theta}-\frac{1}{\kappa \abs{\alpha}^2\mathtt{CRLB}_{\theta}^0}\Big),
	\end{align}
	which is convex. Similarly, we can rewrite \eqref{eq::probMRT::cons_CRB_phi} as
	\begin{align}\label{eq_cons_CRB_phi_SOC}
		&\Big\|\Big[\xibf^\T \bm{\gamma} \big(\Nr \adottheta^\H \adotphi  + \Nt \bdottheta^\H \bdotphi\big) + \rho \Nt^2 \bdottheta^\H \bdotphi;\; 0.5\Big(\varphi_{\theta} - \varpi_{\phi}\nonumber\\
		&\!+\!\frac{1}{\kappa \abs{\alpha}^2\mathtt{CRLB}_{\phi}^0}\Big)\Big]\Big\|_2 \leq 0.5\Big(\varphi_{\theta} \!+\! \varpi_{\phi}\!-\!\frac{1}{\kappa \abs{\alpha}^2\mathtt{CRLB}_{\phi}^0}\!\Big).
	\end{align}
	
	Thus, the approximate convex problem associated with \eqref{ori_problem_MRTZF} to be solved at iteration $i$ is given by
	\begin{subequations}\label{Convex_problem_MRTZF}
		\begin{IEEEeqnarray} {rcl}
		(\mathcal{P}_{\mathtt{bf}}^{\text{SOCP}}): \quad & \underset{\substack{ \bm{\gamma}, \rho }}{\textrm{maximize}}  & \quad  \sum_{k=1}^K {R_{\mathtt{bf}}^{(i)}}_k(\bm{\gamma},\rho) \label{eq_objectiveconvex} \\
			&\textrm{subject to} 
			&\quad \eqref{eq::probMRT::cons_power_MRT}, \eqref{eq_cons_CRB_theta_SOC}, \eqref{eq_cons_CRB_phi_SOC}.\qquad
		\end{IEEEeqnarray}
	\end{subequations}

	\begin{algorithm}[t]
		\begin{algorithmic}[1]
			\fontsize{10}{10}\selectfont
			\protect\caption{Proposed Iterative Algorithm for Solving  \eqref{ori_problem_MRTZF} with MRT or ZF-Based Transmission}
			\label{alg1}
			\global\long\def\algorithmicrequire{\textbf{Initialization:}}
			\REQUIRE  Set $i:=1$ and  randomly generate an initial feasible point for $(\bm{\gamma}^{(0)},\rho^{(0)})$.
			\REPEAT
			\STATE Solve  \eqref{Convex_problem_MRTZF} to obtain the optimal transmission power $(\bm{\gamma}^{\star},\rho^{\star})$;
			
			\STATE Update:\ \ $(\bm{\gamma}^{(i)},\rho^{(i)}) := (\bm{\gamma}^{\star},\rho^{\star})$;
			\STATE Set $i:=i+1$;
			\UNTIL Convergence\\
			\STATE{\textbf{Output:}} $(\bm{\gamma}^{\star},\rho^{\star})$.
	\end{algorithmic} \end{algorithm}

	\subsection{Overall Algorithm and Complexity Analysis}
	The proposed iterative procedure to solve \eqref{ori_problem_MRTZF} is summarized in Algorithm \ref{alg1}. At each iteration, the optimal solution of the convex program \eqref{Convex_problem_MRTZF} is considered as the feasible point for the next iteration. This procedure is repeated until convergence.
	
In the convex problem \eqref{Convex_problem_MRTZF}, only the objective \eqref{eq_objectiveconvex} is approximated using the IA framework. The convergence of Algorithm \ref{alg1} to a local optimum of the original problem \eqref{ori_problem_MRTZF} has been well studied in the literature~\cite{Beck:JGO:10,Dinh:TCOMM:2017,Dinh:JSAC:Dec2017}. In particular, it can be shown that ${R_{\mathtt{bf}}}_k(\bm{\gamma},\rho) \geq{R_{\mathtt{bf}}^{(i)}}_k(\bm{\gamma},\rho)$ and ${R_{\mathtt{bf}}}_k(\bm{\gamma}^{(i)},\rho^{(i)})={R_{\mathtt{bf}}^{(i)}}_k(\bm{\gamma}^{(i)},\rho^{(i)})$. Based on~\cite{Beck:JGO:10}, ${R_{\mathtt{bf}}}_k^{(i+1)}(\bm{\gamma},\rho) \geq {R_{\mathtt{bf}}^{(i)}}_k(\bm{\gamma},\rho)$, and the equality holds whenever $(\bm{\gamma}^{(i+1)},\rho^{(i+1)})\equiv(\bm{\gamma}^{(i)},\rho^{(i)})$. This implies that Algorithm \ref{alg1} will produce a sequence of non-decreasing objective values that converge to at least a local optimum for a sufficiently large number of iterations. Problem \eqref{Convex_problem_MRTZF} has only three SOC and linear constraints and $K+1$ scalar decision variables. As a result, the worst-case computational complexity per iteration of Algorithm \ref{alg1} is $\mathcal{O}(\sqrt{3}K^3)$ \cite[Chapter 6]{Ben:2001}.

	\section{Simulation Results}
	\label{sec_sim}
	\subsection{Simulation Setup}
	In this section, we provide numerical results to validate the theoretical findings and proposed design. We consider a scenario where the users have random locations that are uniformly distributed within a cell of radius of $1000$ meters (m) with the BS is at the center, except that no user is closer to the BS than $r_{\mathtt{h}}$ = $100$ m. Furthermore, we assume that the target is located $400$ m away from the BS with angles $(\theta,\phi) = \left(\frac{\pi}{8}, \frac{\pi}{4}\right)$. The large-scale fading parameters are computed as $\beta_k = z_k/(r_k/r_{\mathtt{h}})^{\nu}$, where $z_k$ is a log-normal random variable with standard deviation $\sigma_{\mathtt{shadow}} = 7$ dB, $r_k$ is the distance between the $k$-th user and the BS, and $\nu = 3.2$ is the path loss exponent~\cite{ngo2013energy}. For simplicity, we assume $\alpha = \frac{1}{\beta_{\mathtt{s}}\sqrt{2}}(1 + j)$, where $\beta_{\mathtt{s}}$ is the round-trip path loss between the BS and the target. 
 
    We set $L=30$, and define the SNR as $\mathrm{SNR} = P_{\mathtt{t}}$ {since the noise variances are normalized as $\sigma_c = \sigma_s = \noise = 1$ for all users  \cite{ngo2013energy}.} During the uplink training phase, we employ $\tau_{\mathtt{p}} = 10$, $\tau_{\mathtt{c}} = 100$, $\mathrm{SNR}  = 30$ dB, and mutually orthogonal pilot sequences with a pilot reuse pattern~\cite{van2021reconfigurable}
    \begin{align*}
        \vp_k^\H \vp_j = 
        \begin{cases}
            1, \text{~if~} j \in  \mathcal{P}_k\\
            0, \text{~otherwise},
        \end{cases}
    \end{align*}
    where $\mathcal{P}_k$ is the set of user indices (including user $k$) that share the same pilot sequence as user $k$. For the array response vectors in \eqref{eq_ath} and \eqref{eq_atv}, we choose $\Nth = \Ntv = \sqrt{\Nt} \in \mathbb{N}$ and $\Nrh = \Nrv = \sqrt{\Nr} \in \mathbb{N}$.

	We employ CVX to solve problem $\mathcal{P}_{\mathtt{bf}}^{\text{SOCP}}$ in \eqref{Convex_problem_MRTZF} for the proposed power allocation scheme. For comparison, we consider four benchmark approaches as follows:
 
    $\bullet$ Equal power fractions among communications users (\textit{EqualCom}): In this scheme, only the power allocation between communications and sensing is performed, while the fraction of power assigned to the all communications users is equal. This solution is obtained by solving problem \eqref{ori_problem_MRTZF} with the additional constraint $\gamma_1 = \gamma_2 = \ldots = \gamma_K$.  
    
    $\bullet$ Equal power allocation between communications and sensing as well as among communications users (\textit{EqualC\&S}): In this scheme, the available power budget is first shared equally between communications and sensing, each with $\Pt/2$, and then we set $\gamma_1 = \gamma_2 = \ldots = \gamma_K$. From \eqref{eq_power}, we can determine $\left(\rho,\gamma_k\right) = \left( \frac{\Pt}{2\Nt},  \frac{\Pt}{2\Nt \sum_{k=1}^K {\xi_{\mathtt{bf}}}_k} \right)$. This setting achieves the equality in \eqref{eq::probMRT::cons_power_MRT}, i.e., all the power budget is utilized for the joint transmission.

    Simulation results for the rates and CRLB of these benchmarks are obtained for both the MRT and ZF schemes, i.e., with $\mathtt{bf} = \{\mathtt{MRT},\mathtt{ZF}\}$ in \eqref{eq_SE_theo}, \eqref{eq_CRB_theta_1}, and \eqref{eq_CRB_phi_1}. As a result, we have four benchmarks in our simulations: \textit{MRT/ZF-EqualCom} and \textit{MRT/ZF-EqualC\&S}. Unless otherwise stated, the subsequent simulation results are obtained by averaging over $10$ sets of large-scale fading parameters and $100$ small-scale channel realizations. {In the figures, we show the CRLB in dB as follows: $\mathtt{CRLB~[\mathrm{dB}]} = 10 \log_{10}(\mathtt{CRLB~[\mathrm{rad}^2]})$.}

	\subsection{Convergence of Algorithm \ref{alg1}}
	\begin{figure}[t]%
		\vspace{-0.5cm}\centering
		\includegraphics[scale=0.52]{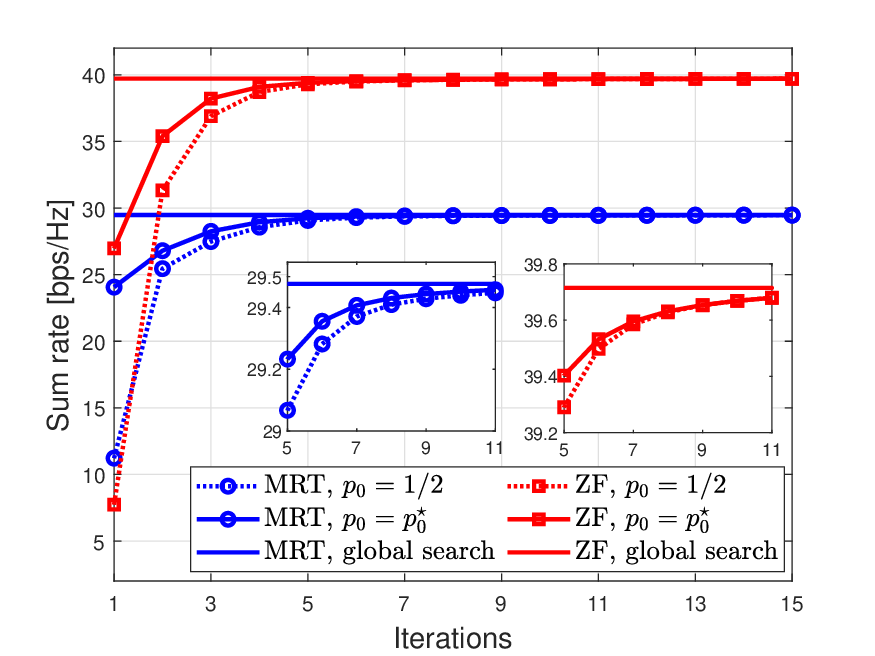}
		\caption[]{{Convergence of Algorithm \ref{alg1} with $\Nt=225$, $\Nr=25$, $K=12$, $L=30$, $\mathtt{CRLB}_{\theta}^0 = \mathtt{CRLB}_{\phi}^0 = -35$ dB, $\mathrm{SNR} =10$ dB. The objective value of the global search method is achieved after $1000$ s with $14$ parallel solvers.}}%
		\label{fig_conv}%
	\end{figure}

    \renewcommand{\arraystretch}{1.0}

    \begin{figure}[t]
            \vspace{-0.5cm}
        \centering
        \includegraphics[scale=0.52]{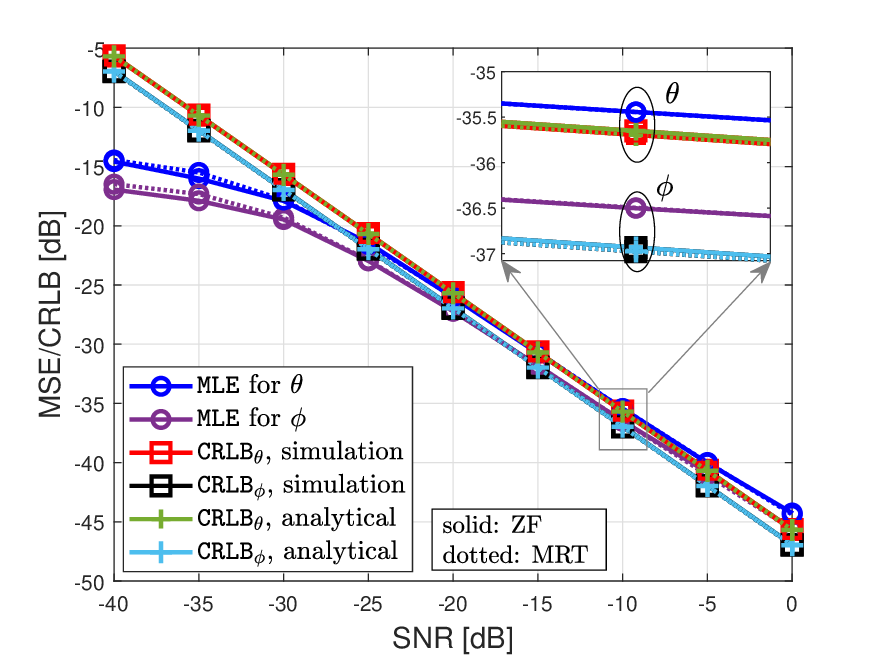}
        \caption{{CRLB and MSE of the MLE for $\theta$ and $\phi$ with $\Nt=25$, $\Nr=25$, $K=8$, $L=30$. The step size for the MLE grid search is $\pi/256$.}}
        \label{fig_CRB}
    \end{figure}
    
	In Fig.\ \ref{fig_conv} we show the convergence of Algorithm \ref{alg1}. We consider both MRT and ZF with $N = 250$, $\Nt = 225$ and $\Nr = N - \Nt = 25$, $K=12$, $\mathtt{CRLB}_{\theta}^0 = \mathtt{CRLB}_{\phi}^0 = -35$ dB, and $\mathrm{SNR} = 10$ dB. The algorithm is initialized with $\left(\rho^{(0)},\gamma_k^{(0)}\right) = \left( \frac{p_0 \Pt}{\Nt},  \frac{(1-p_0)\Pt}{\Nt \sum_{k=1}^K {\xi_{\mathtt{bf}}}_k} \right)$, which satisfies the power constraint \eqref{eq::probMRT::cons_power_MRT}. In Fig.\ \ref{fig_conv}, we show the convergence for $p_0 = \frac{1}{2}$ and that obtained via a search to find the smallest $p_0$, denoted as $p_0^{\star}$, satisfying \eqref{eq::probMRT::cons_CRB_theta} and \eqref{eq::probMRT::cons_CRB_phi}. Note that with $p_0 = \frac{1}{2}$, the initial solution is the same as for EqualC\&S, which meets the CRLB constraints with a high probability because of the large power fraction allocated for sensing. We see from Fig.\ \ref{fig_conv} that for both initialization methods, Algorithm \ref{alg1} converges after a few iterations and achieves similar objective values at convergence.  However, the initialization with $p_0 = p_0^{\star}$ offers much better initial values and faster convergence for both MRT and ZF, compared to that obtained by fixing $p_0 = \frac{1}{2}$. This is because in this case, the smallest power fraction is allocated for sensing, while the remaining power is used for communications. Note that searching for $p_0^{\star}$ requires very low complexity because the cost of computing the CRLBs every time $(\bm{\gamma}, \rho)$ changes is only $\mathcal{O}(K)$. 
    
    {In Fig.\ \ref{fig_conv}, we compare the sum rates achieved by Algorithm~\ref{alg1} to those obtained using MATLAB's built-in global search method. We employ $14$ local solvers running in parallel for $1000$ seconds (s) for the global search, ensuring a high probability of identifying the global optimum of problem \eqref{ori_problem_MRTZF}. As shown in Fig.\ \ref{fig_conv}, the proposed scheme achieves sum rates that closely approximate those of the global search method with significantly reduced runtime. Specifically, Algorithm~\ref{alg1} completes $15$ iterations in just $5.90$ s and $7.02$ s for MRT and ZF precoding, respectively.}
    
 

	\subsection{Communications and Sensing Performance }

	In this section, we demonstrate the communications and sensing performance of MRT and ZF precoding with the proposed power allocation. We first verify the theoretical results for the CRLB in Theorems \ref{theo_SE} and Remark \ref{rm_CRB_0}. In Fig.\ \ref{fig_CRB}, we compare the CRLB for $\theta$ and $\phi$ obtained by Monte-Carlo simulations based on \eqref{eq_CRB_theta_0}, \eqref{eq_CRB_phi_0}, and the analytical results in Remark \ref{rm_CRB_0}.  We also show the MSE of the MLE in \eqref{eq_MLE} for $\theta$ and $\phi$, which performs an exhaustive search over a fine grid. The sensing SNR is defined as $\frac{\Pt L \abs{\alpha}^2}{\sigmas}$~\cite{liu2021cramer}. It is observed for both MRT and ZF that the CRLB is a tight lower bound for the MSE of MLE. Therefore, the power allocation constraint on the CRLB threshold can guarantee the sensing performance in terms of the angle estimation accuracy.
	
	\begin{figure*}[!htb]
            \vspace{-0.5cm}
		\hspace{-0.5cm}
		\subfigure[Communications performance]
		{
			\includegraphics[scale=0.45]{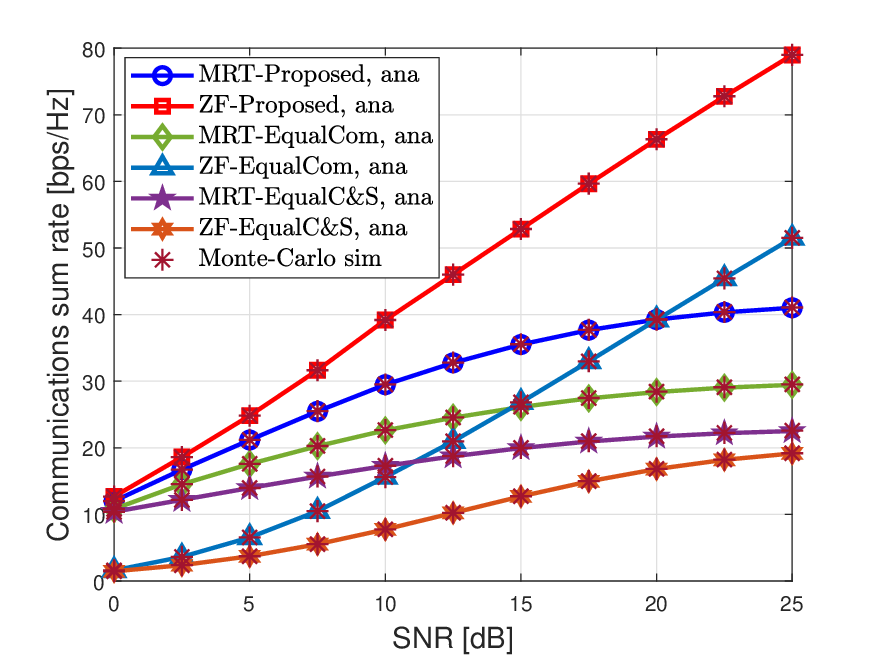}
			\label{fig_rate_SNR}
		}\hspace{-0.8cm}
		\subfigure[Sensing performance]
		{
			\includegraphics[scale=0.45]{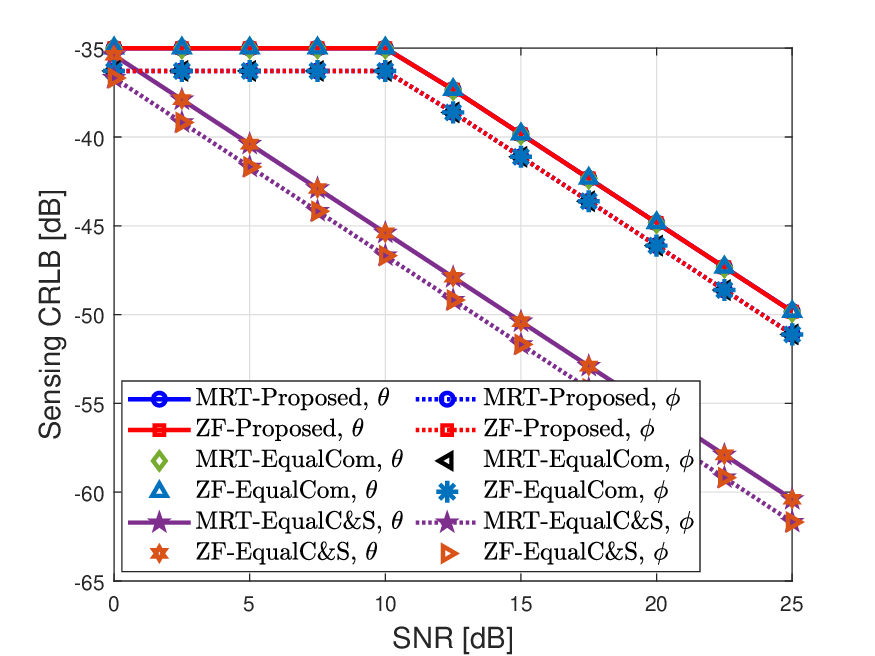}
			\label{fig_CRB_SNR}
		}\hspace{-0.8cm}
		\subfigure[Communications-sensing tradeoff]
		{
			\includegraphics[scale=0.45]{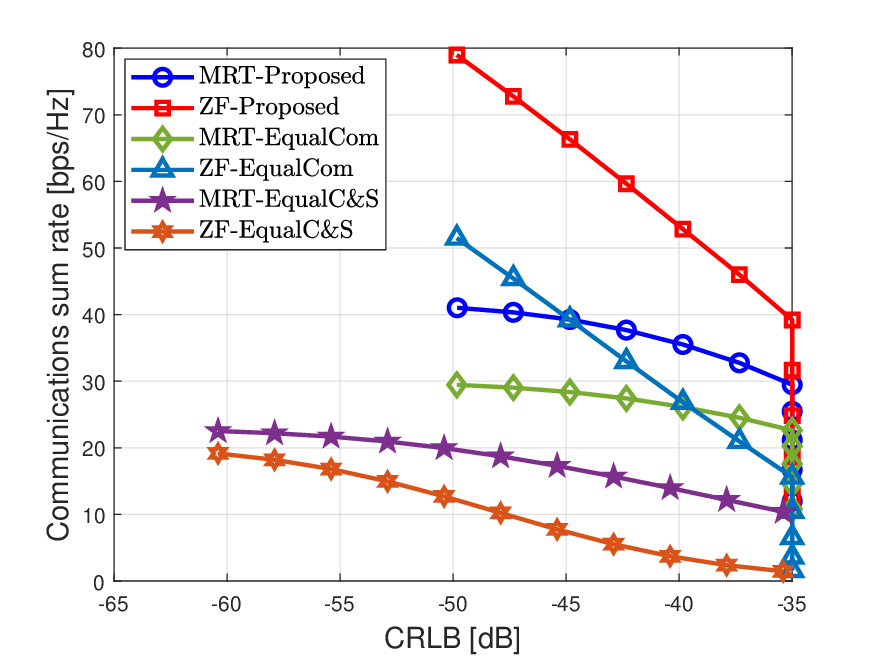}
			\label{fig_tradeoff_SNR}
		}
		\caption{Communications sum rate and sensing CRLBs of MRT and ZF with $\Nt=225$, $\Nr=25$, $K=12$, $L=30$, and $\mathtt{CRLB}_{\theta}^0 = \mathtt{CRLB}_{\phi}^0 = -35$ dB.}
		\label{fig_perf_SNR}
	\end{figure*}
	
	\begin{figure}[t]%
		\vspace{-0.5cm}\centering
		\includegraphics[scale=0.52]{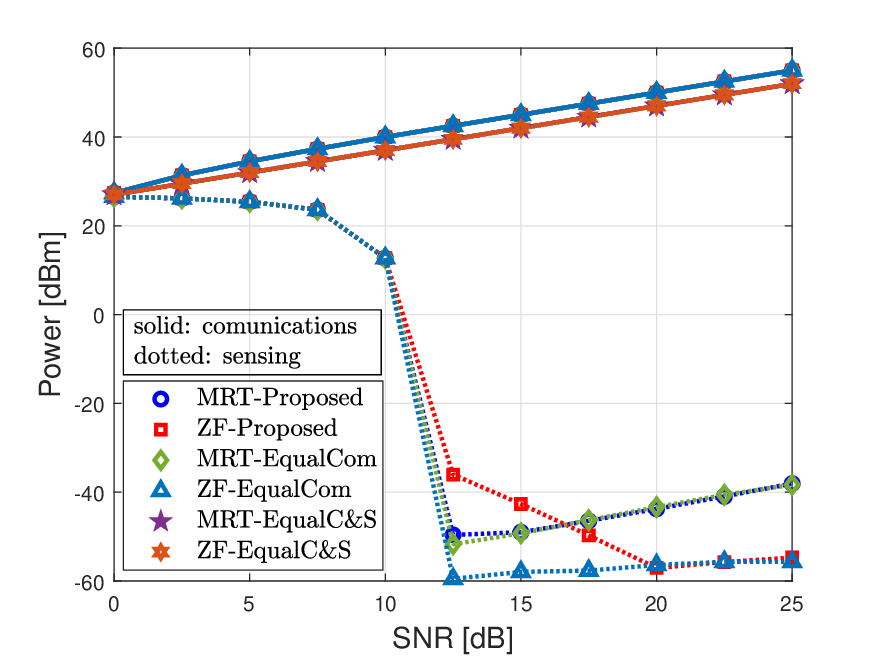}
		\caption[]{Communications and sensing power for the results in Fig.\ \ref{fig_perf_SNR}.}%
		\label{fig_power_SNR}%
	\end{figure}

	In Fig.\ \ref{fig_perf_SNR}, we show the achievable sum rates and the CRLB attained by the proposed power allocation scheme compared with the benchmarks. We set $N = 250$, $\Nt = 225$, $\Nr = 25$, $K=12$, $\mathtt{CRLB}_{\theta}^0 = \mathtt{CRLB}_{\phi}^0 = -35$ dB. We note the following observations:
	\begin{itemize}
		\item The sum rates shown in Fig.\ \ref{fig_rate_SNR} are obtained based on both the closed form expressions in Theorem \ref{theo_SE} and Monte-Carlo simulations. The analytical and simulation results match well for all considered scenarios, validating our theoretical analysis 
        in~Theorem \ref{theo_SE}.

		\item In Fig.\ \ref{fig_rate_SNR}, ZF-Proposed offers the highest sum rate among the ZF-based approaches, followed by ZF-EqualCom and EqualC\&S. For example, at $\mathrm{SNR} =20$ dB, the rate for ZF-Proposed is $294\%$ and $68.9\%$ higher than that of ZF-EqualC\&S and ZF-EqualCom, respectively. ZF-EqualCom with only power allocation between communications and sensing performs $133.3\%$ better than EqualC\&S. These observations demonstrate the significant performance gain from the proposed power allocation method. In addition, the MRT-Proposed algorithm also significantly outperforms the other MRT-based counterparts, although the gap is not as great as for ZF.
  
        \item We see from Fig.\ \ref{fig_CRB_SNR} that the EqualC\&S methods achieve the lowest CRLB because they use the most power for sensing (equal to the amount for communications). Note that the larger sensing power leads to a larger $\Nt \beta_k \rho$ in \eqref{eq_SE_theo}. Thus, EqualC\&S exhibits a significant communications performance loss compared with the Proposed and EqualCom approaches, as discussed in Remark \ref{rm_effect_of_sensing}. ZF-Proposed and ZF-EqualCom have identical CRLBs because they share the same values for $\bm{\xi}^\T \bm{\gamma}$ and $\rho$, as noted in Remark \ref{rm_CRB_2}. Similar observations and discussions can be made for the MRT-based schemes. 
		
		
		\item As seen in Fig.\ \ref{fig_tradeoff_SNR}, ZF-Proposed achieves the best communications-sensing performance tradeoff, far better than MRT-based methods. The EqualC\&S approach yields a very poor tradeoffs although it uses all the power budget. This emphasizes the significance of power allocation in the considered system.
	\end{itemize}

    {In Fig.~\ref{fig_power_SNR}, we illustrate the power allocated to the communications and sensing subsystems, represented by the first and second terms in \eqref{eq_power}, respectively. These results are obtained from the simulation corresponding to Fig.~\ref{fig_perf_SNR}. At low SNRs, a significant amount of power is required for sensing to satisfy the CRLB constraints. As the SNR increases, more power is allocated to communications to maximize the sum rate. Simultaneously, the communications power also contributes to the sensing performance, leading to a decrease in the CRLB, as observed in \eqref{eq_CRB_theta_1}, \eqref{eq_CRB_phi_1}, and Fig.~\ref{fig_CRB_SNR}. Consequently, as the SNR increases, the sensing power decreases significantly. When the communications power becomes sufficiently large, only a minimal amount of power is allocated to sensing, as observed for SNR $> 10$~dB in Fig.~\ref{fig_power_SNR}. It is noted that the sensing power slightly increases in this SNR regime, but the power values are extremely small, i.e., below $-40$~dBm, and are entirely dominated by the communications power. As a result, the effect of sensing power on the communications and sensing functions in this regime is negligible.
}
        
	In Fig.\ \ref{fig_perf_CRB0}, we show the communications sum rate and power versus the CRLB thresholds. We consider $\mathtt{CRLB}_{\theta}^0 = \mathtt{CRLB}_{\phi}^0 \in [-55,-20]$ dB and $\mathrm{SNR} =20$ dB; the other simulation parameters are the same as those in Fig.\ \ref{fig_perf_SNR}. The sum rates of the proposed and EqualCom schemes increase with the CRLB threshold in the range $[-55,-45]$ dB. Beyond that they remain nearly unchanged. This is reasonable and aligns well with the power allocation in Fig.\ \ref{fig_power_CRB0}. More specifically, the CRLB constraints \eqref{eq::probMRT::cons_CRB_theta} and \eqref{eq::probMRT::cons_CRB_phi} are easier to satisfy when the threshold increases. Therefore, the power allocated for sensing decreases quickly to allow more for communications to maximize the sum rate. Note that the sum rates of the two EqualC\&S schemes remain constant because the communications and sensing power are fixed and independent of the CRLB thresholds.

    \begin{figure}[t]
            \vspace{-0.5cm}
		\centering
		\subfigure[Communications performance]
		{
            \includegraphics[scale=0.52]{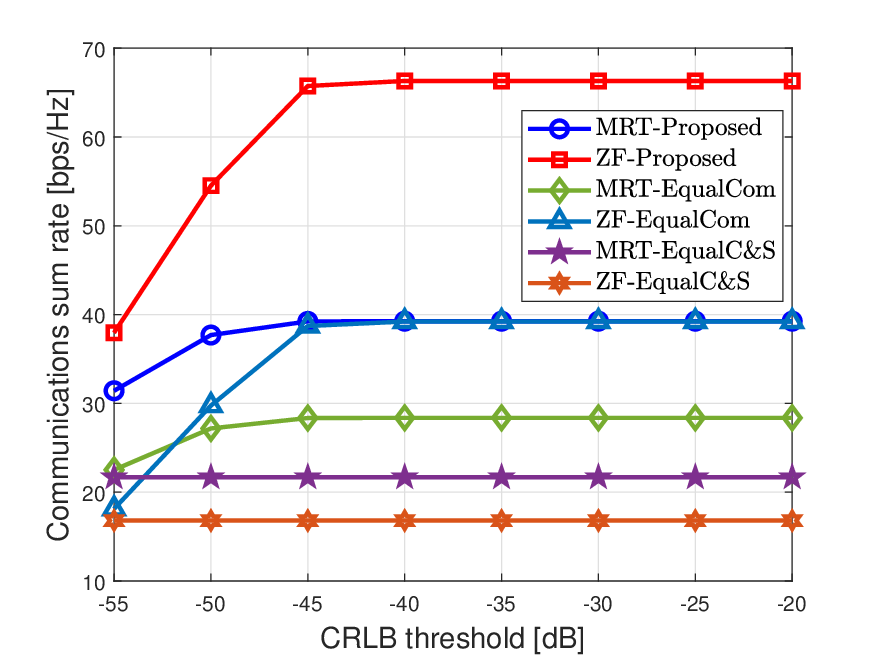}\vspace{-0.35cm}
			\label{fig_rate_CRB0t}
		}\vspace{-0.35cm}
		\subfigure[Power for communications and sensing]
		{
			\includegraphics[scale=0.52]{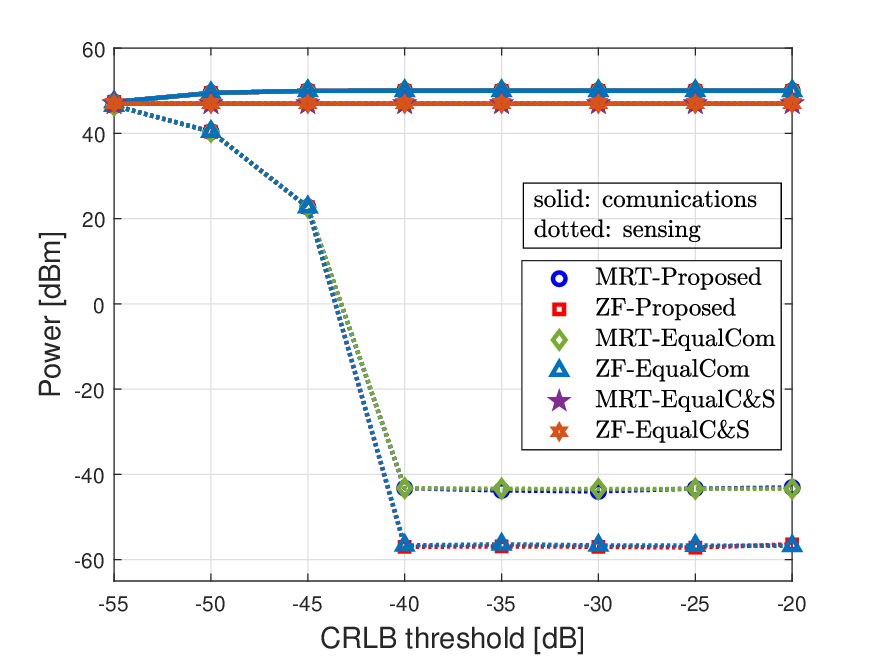}
			\label{fig_power_CRB0}
		}
		\caption{Sum rate and power with  $\mathtt{CRLB}_{\theta}^0 = \mathtt{CRLB}_{\phi}^0 = [-55,-20]$ dB $\Nt=225$, $\Nr=25$, $K=12$, $L=30$, $\mathrm{SNR} = 20$ dB.}
		\label{fig_perf_CRB0}
	\end{figure}
    
	\begin{figure}[t]
		\vspace{-0.5cm}\centering
		\subfigure[Communications performance]
		{
			\includegraphics[scale=0.52]{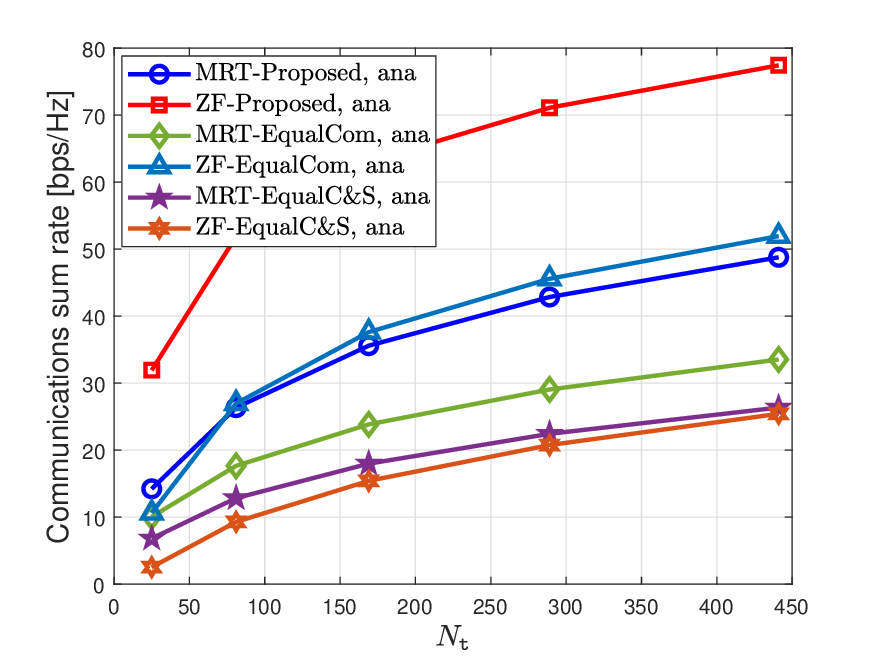}\vspace{-0.35cm}
			\label{fig_rate_Nt}
		}\vspace{-0.35cm}
		\subfigure[Sensing performance]
		{
			\includegraphics[scale=0.52]{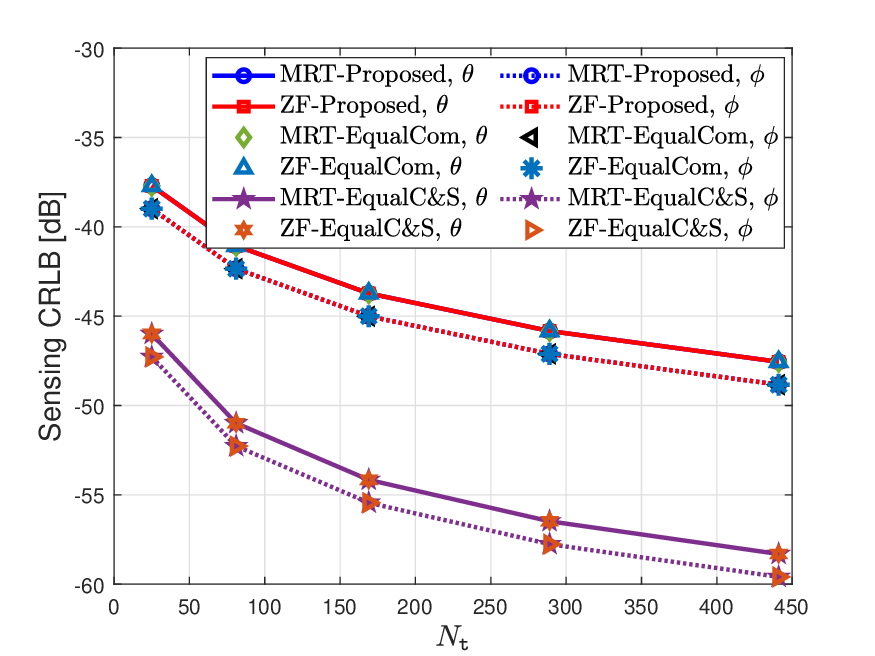}
			\label{fig_CRB_Nt}
		}
		\caption{Communications sum rate and sensing CRLB for $\Nr=25$, $K=12$, $L=30$, $\mathtt{CRLB}_{\theta}^0 = \mathtt{CRLB}_{\phi}^0 = -35$ dB, and $\mathrm{SNR} =20$ dB.}
		\label{fig_perf_Nt}
	\end{figure}

	Finally, we show the communications sum rates and sensing CRLBs when increasing $\Nt$, while the other parameters are set the same as those in Fig.\ \ref{fig_perf_SNR}. Fig.\ \ref{fig_rate_Nt} shows that as $\Nt$ increases, the sum rates of all approaches increase significantly. This is because the inter-user interference is mitigated with large $\Nt$, as discussed in Remark \ref{rm_free_interference}. The increase in the sum rate of ZF-EqualC\&S is slower than that for ZF-Proposed and ZF-EqualCom since the term $\Nt \beta_k \rho$ also increases with $\Nt$, which affects the communications performance of ZF-EqualC\&S without proper power allocation. In terms of sensing performance, the CRLBs of all considered approaches decrease with $\Nt$ and tend to zero as $\Nt \rightarrow \infty$. This verifies the conclusion in Remark \ref{rm_CRB_1}.
	

	\section{Conclusion}
	\label{sec_conclusion}
	We have investigated mono-static multiuser massive MIMO ISAC systems with linear MRT and ZF precoding. To characterize the system performance and operational properties, we derived closed form expressions for the achievable communications rate and sensing CRLBs. The analytical findings reveal important properties about ISAC operations in massive MIMO scenarios, such as the mutual interference among the subsystems and the impact of the large number of antennas on communications and sensing performance. We proposed an algorithm for power allocation among the precoders for the communications users and the sensed target to maximize the users' sum rate with a constraint on the CRLB and power budget. Our theoretical findings and proposed algorithm were verified by extensive numerical results, which show superior communications and sensing performance in massive MIMO ISAC systems. 

	\appendices

        \section{Proof of Lemma \ref{lemma_power}}
	\label{appd_power}
	With MRT, the total transmit power is computed as
	\begin{align*}
		&P_{\mathtt{bf}}^{\mathtt{MRT}} = \mean{\tr{\mF_{\mathtt{MRT}} \mF_{\mathtt{MRT}}^\H}} = \tr{\mean{\mF_{\mathtt{MRT}} \mF_{\mathtt{MRT}}^\H}} \\
		& = \text{trace}\left( \mathbb{E} \left\{ \left(\hat{\mH} \bm{\Gamma} + {\vv} \bar{\bm{\eta}}^\T\right) \left(\hat{\mH} \bm{\Gamma} + {\vv} \bar{\bm{\eta}}^\T\right)^\H \right\} \right)\\
		&= \underbrace{\mean{\tr{\hat{\mH} \bm{\Gamma} \bm{\Gamma}^\H \hat{\mH}^\H}}}_{\triangleq \varepsilon_{\mathtt{MRT}}^{\mathtt{tot}}} + \underbrace{\mean{\tr{{\vv} \bar{\bm{\eta}}^\T \bar{\bm{\eta}} {\vv}^\H }}}_{\triangleq \varepsilon_0^{\mathtt{tot}}}. \nbthis \label{eq_mean_01}
	\end{align*}
	First, note that $ \bm{\Gamma}$ is a diagonal matrix. We have
	\begin{align*}
		\varepsilon_{\mathtt{MRT}}^{\mathtt{tot}} &= \mean{\tr{\bm{\Gamma}^2 \hat{\mH}^\H \hat{\mH}}} = \Nt \bm{\xi}_{\mathtt{MRT}}^\T \bm{\gamma}, \nbthis \label{eq_mean_10}
	\end{align*}
	where the last equality follows from the fact that the mean of the $k$-th diagonal element of $\hat{\mH}^\H \hat{\mH}$ is equal to $\Nt \xi_k$. Furthermore, recall that ${\vv}^\H {\vv} = \Nt$, we have
	\begin{align*}
		\varepsilon_0^{\mathtt{tot}} 
            = \mean{\tr{\bar{\bm{\eta}}^\T \bar{\bm{\eta}}}} = \Nt \sum\nolimits_{k=1}^K \eta_k = \Nt \rho. \nbthis \label{eq_mean_21}
	\end{align*}
	From \eqref{eq_mean_01}--\eqref{eq_mean_21}, the total transmit power constraint in the case of MRT can be given as \eqref{eq_power} with $\xibf = \bm{\xi}_{\mathtt{MRT}}$.

	With ZF precoding, the total transmit power is computed as
	\begin{align*}
		&P_{\mathtt{bf}}^{\mathtt{ZF}} = \mean{\tr{\mF_{\mathtt{ZF}} \mF_{\mathtt{ZF}}^\H}} = \tr{\mean{\mF_{\mathtt{ZF}} \mF_{\mathtt{ZF}}^\H}} \\
		& = \text{trace}\left( \mathbb{E} \left\{ \left(\hat{\mH}^{\dagger} \bm{\Gamma} + {\vv} \bar{\bm{\eta}}^\T\right) \left(\hat{\mH}^{\dagger} \bm{\Gamma} + {\vv} \bar{\bm{\eta}}^\T\right)^\H \right\} \right)\\
		&= \underbrace{\mean{\tr{\bm{\Gamma} \bm{\Gamma}^\H (\hat{\mH}^{\dagger})^\H \hat{\mH}^{\dagger} }}}_{\triangleq \varepsilon_{\mathtt{ZF}}^{\mathtt{tot}}}  + \underbrace{\mean{\tr{{\vv} \bar{\bm{\eta}}^\T \bar{\bm{\eta}} {\vv}^\H }}}_{= \varepsilon_0^{\mathtt{tot}}}. \nbthis \label{eq_mean_ZF}
	\end{align*}
    Note that $ \bm{\Gamma}$ is a diagonal matrix of real entries, which are independent of the small-scale fading channel coefficients. Furthermore, we can express $\hat{\mH}$ as $\hat{\mH} = \mZ \bm{\Xi}^{\frac{1}{2}}$, where $\mZ = [\vz_1, \ldots, \vz_k]$, $\vz_k \sim \mathcal{CN}(\bm{0}, \mI_{\Nt})$, and $\bm{\Xi} = \diag{\xi_1, \ldots, \xi_K}$.  We can expand
	\begin{align*}
		&\varepsilon_{\mathtt{ZF}}^{\mathtt{tot}} = \mean{\tr{\bm{\Gamma}^2 (\hat{\mH}^\H \hat{\mH})^{-1}}} \\
		&=\! \mean{\tr{\bm{\Gamma}^2 \bm{\Xi}^{-1} (\mZ^\H \mZ)^{-1}}} 
		\!=\! \sum_{k=1}^{K} \frac{\gamma_k}{\xi_k} \mean{\left[ (\hat{\mZ}^\H \hat{\mZ})^{-1} \right]_{kk}}\\
		& =\! \sum_{k=1}^{K}\!\! \frac{\gamma_k}{K\xi_k} \mean{\tr{(\hat{\mZ}^\H \hat{\mZ})^{-1}}} \! = \!\sum_{k=1}^{K} \frac{\gamma_k}{(\Nt-K)\xi_k}, \nbthis \label{eq_e_tot_ZF}
	\end{align*}
	where we have used the property $\meanshort{\trshort{(\hat{\mZ}^\H \hat{\mZ})^{-1}}} = \frac{K}{\Nt-K}$, since $\hat{\mZ}^\H \hat{\mZ}$ is a central complex Wishart matrix~\cite{yang2013performance}. From \eqref{eq_mean_21}--\eqref{eq_e_tot_ZF}, the total power constraint with the ZF precoder is given as \eqref{eq_power} where $\xibf = \bm{\xi}_{\mathtt{ZF}}$ as defined in \eqref{eq_def_xi_bf}.

	\section{Proof of Theorem \ref{theo_SE}}
	\label{appd_SE}
	
	We first rewrite the terms in \eqref{eq_ykl_MRT} as $\mathtt{DS}_k = \mean{ \vh_k^\H {\vf_{\mathtt{bf}}}_k}, \mathtt{BU}_k = \vh_k^\H {\vf_{\mathtt{bf}}}_k - \mean{ \vh_k^\H {\vf_{\mathtt{bf}}}_k}, \mathtt{UI}_{kj} = \vh_k^\H  {\vf_{\mathtt{bf}}}_j$,
	where $\mathtt{bf} \in \{\mathtt{MRT}, \mathtt{ZF} \}$. Here,
	\begin{align*}
		&\vh_k^\H {\vf_{\mathtt{bf}}}_i = (\hat{\vh}_k + \ve_k)^\H \left(\sqrt{\gamma_i} {\vw_{\mathtt{bf}}}_i + \sqrt{\eta_i} {\vv}\right)\\
		&= \sqrt{\gamma_i} \hat{\vh}_k^\H {\vw_{\mathtt{bf}}}_i + \sqrt{\eta_i}  \hat{\vh}_k^\H {\vv}  + \ve_k^\H (\sqrt{\gamma_i} {\vw_{\mathtt{bf}}}_i + \sqrt{\eta_i} {\vv}), \nbthis \label{eq_hkfmrti}
	\end{align*}
	for $k,i = 1, \ldots, K$, and ${\vw_{\mathtt{bf}}}_i$ is the $i$-th column of $\mW$ in \eqref{eq_W_linear}. Note that
	\begin{align*}
		\mean{\sqrt{\eta_i}  \hat{\vh}_k^\H {\vv}  + \ve_k^\H (\sqrt{\gamma_i} {\vw_{\mathtt{bf}}}_i + \sqrt{\eta_i} {\vv})} = 0, \forall k, i \nbthis \label{eq_mean_0}
	\end{align*}
	as $\ve_k$ and ${\vw_{\mathtt{bf}}}_i$ are independent  and they both have zero means.
	To derive the closed form SE expressions, in the following we compute $\abs{\mathtt{DS}_k}^2$, $\meanshort{\abs{\mathtt{BU}_k}^2}$, and $\meanshort{\abs{\mathtt{UI}_{kj}}^2}$.

        \vspace{-0.25cm}
	\subsection{MRT Beamforming}
	
	With MRT beamforming, we have ${\vw_{\mathtt{bf}}}_i = \hat{\vh}_i$.
	
	\subsubsection{Computation of $\abs{\mathtt{DS}_k}^2$}
	Using \eqref{eq_hkfmrti} and \eqref{eq_mean_0}, we have
	\begin{align*}
		\abs{\mathtt{DS}_k}^2 &= \gamma_k \abs{\mean{\normshort{\hat{\vh}_k}^2}}^2 = \Nt^2 \xi_k^2 \gamma_k . \nbthis \label{eq_abs_DSk2_MRT}
	\end{align*}
	
	\subsubsection{Computation of $\meanshort{\abs{\mathtt{BU}_k}^2}$}
	We have
	\begin{align*}
		\mean{\abs{\mathtt{BU}_k}^2} &= \mean{\abs{\vh_k^\H {\vf_{\mathtt{MRT}}}_k - \mean{ \vh_k^\H {\vf_{\mathtt{MRT}}}_k}}^2}\\
		&= \underbrace{\mean{\abs{\vh_k^\H {\vf_{\mathtt{MRT}}}_k}^2}}_{\triangleq E_0} - \underbrace{\abs{\mean{ \vh_k^\H {\vf_{\mathtt{MRT}}}_k}}^2}_{= \abs{\mathtt{DS}_k}^2}.\nbthis \label{eq_mean_abs_BUk2_MRT}
	\end{align*}
	Here, $E_0$ can be computed as:
	\begin{align*}
		E_0 &= \mean{\abs{\sqrt{\gamma_k} \normshort{\hat{\vh}_k}^2 + \sqrt{\eta_k}  \hat{\vh}_k^\H {\vv}  + \ve_k^\H (\sqrt{\gamma_k} \hat{\vh}_k + \sqrt{\eta_k} {\vv})}^2}\\
		&= \gamma_k \mean{\normshort{\hat{\vh}_k}^4} + \eta_k \mean{\absshort{\hat{\vh}_k^\H {\vv}}^2} \\
		&\hspace{2cm} + \mean{\absshort{\ve_k^\H (\sqrt{\gamma_k} \hat{\vh}_k + \sqrt{\eta_k} {\vv})}^2}. \nbthis \label{eq_E000}
	\end{align*}
	Note that with $\vv = \va(\tilde{\theta}, \tilde{\phi})$, the elements of ${\vv}$ are deterministic with unit modulus, and $\hat{\vh}_k \sim \mathcal{CN}(0, \xi_k \mI_{\Nt})$. Thus, 
	\begin{align*}
		&\mean{\normshort{\hat{\vh}_k}^4} = \Nt (\Nt+1) \xi_k^2, \nbthis \label{eq_E01}\\
		&\mean{\absshort{\hat{\vh}_k^\H {\vv}}^2} = \mean{\normshort{\hat{\vh}_k}^2} = \Nt \xi_k, \nbthis \label{eq_E02}\\
		&\mean{\absshort{\ve_k^\H (\sqrt{\gamma_k} \hat{\vh}_k + \sqrt{\eta_k} {\vv})}^2} 
        = \Nt \epsilon_k \left(\xi_k \gamma_k + \eta_k \right). \nbthis \label{eq_E03}
	\end{align*}
    As a result, $E_0$ in \eqref{eq_mean_abs_BUk2_MRT} can be obtained as
    \begin{align*}
		E_0 
        &= \Nt^2 \xi_k^2 \gamma_k + \Nt \beta_k \left( \xi_k \gamma_k + \eta_k \right), \nbthis \label{eq_E0001}
	\end{align*}
	where we have used $\xi_k + \epsilon_k = \beta_k$. From \eqref{eq_abs_DSk2_MRT}, \eqref{eq_mean_abs_BUk2_MRT}, and \eqref{eq_E0001}, we obtain
	\begin{align*}
		&\mean{\abs{\mathtt{BU}_k}^2} = \Nt \beta_k  \left( \xi_k \gamma_k + \eta_k \right). \nbthis \label{eq_mean_abs_BUk2_MRT1}
	\end{align*}
	
	\subsubsection{Computation of $\meanshort{\abs{\mathtt{UI}_{kj}}^2}$}
	Using the results in \eqref{eq_E02} and \eqref{eq_E03}, we have
	\begin{align*}
		&\mean{\abs{\mathtt{UI}_{kj}}^2} = \mean{\abs{\vh_k^\H {\vf_{\mathtt{MRT}}}_j}^2}\\
		&= \mean{\abs{\sqrt{\gamma_j} \hat{\vh}_k^\H \hat{\vh}_j + \sqrt{\eta_j}  \hat{\vh}_k^\H {\vv}  + \ve_k^\H (\sqrt{\gamma_j} \hat{\vh}_j + \sqrt{\eta_j} {\vv})}^2} \\
		&= \Nt \beta_k (\gamma_j \xi_j + \eta_j). \nbthis \label{eq_mean_abs_UIk2_MRT}
	\end{align*}
	From \eqref{eq_abs_DSk2_MRT}, \eqref{eq_mean_abs_BUk2_MRT1}, and \eqref{eq_mean_abs_UIk2_MRT}, the SINR term in \eqref{eq_SE_def} for MRT can be obtained as
	\begin{align*}
		{\mathtt{SINR}_{\mathtt{MRT}}}_k = \frac{\Nt^2 \xi_k^2 \gamma_k }{\Nt \beta_k \rho + \Nt {\bm{\zeta}_{\mathtt{MRT}}^\T}_k \bm{\gamma} + \sigmac},
	\end{align*}
	where $\rho$ and ${\bm{\zeta}_{\mathtt{MRT}}}_k$ are defined in Theorem \ref{theo_SE}. As a result, we obtain \eqref{eq_SE_theo} with $\lambdabf_k = {\lambda_{\mathtt{MRT}}}_k$ and $\bm{\zeta}_{\mathtt{bf},k} = {\bm{\zeta}_{\mathtt{MRT}}}_k$.
	
	\subsection{ZF Beamforming}
	\subsubsection{Compute $\abs{\mathtt{DS}_k}^2$}
	Based on \eqref{eq_hkfmrti} and \eqref{eq_mean_0}, we have
	\begin{align*}
		&\abs{\mathtt{DS}_k}^2 = \abs{\mean{\vh_k^\H {\vf_{\mathtt{ZF}}}_k}}^2\\
		&= \abs{\mean{\left(\hat{\vh}_k + \ve_k\right)^\H \left(\sqrt{\gamma_k} \check{\vh}_k + \sqrt{\eta_k} {\vv}\right)}}^2 
		= \gamma_k \abs{\mean{\hat{\vh}_k^\H \breve{\vh}_k}}^2 \\
		&=  \gamma_k \abs{\mean{\left[\hat{\mH}^\H \hat{\mH} (\hat{\mH}^\H \hat{\mH})^{-1}\right]_{kk}}}^2 = \gamma_k, \nbthis \label{eq_abs_DSk2}
	\end{align*}
	where $[\cdot]_{ij}$ denotes the $(i,j)$-th entry of a matrix.
 \smallskip
	\subsubsection{Compute $\meanshort{\abs{\mathtt{BU}_k}^2}$}
	We have
	\begin{align*}
		\mean{\abs{\mathtt{BU}_k}^2} &= \mean{\abs{\vh_k^\H {\vf_{\mathtt{ZF}}}_k - \mean{ \vh_k^\H {\vf_{\mathtt{ZF}}}_k}}^2}\\
		&= \underbrace{\mean{\abs{\vh_k^\H {\vf_{\mathtt{ZF}}}_k}^2}}_{\triangleq E_0} - \underbrace{\abs{\mean{ \vh_k^\H {\vf_{\mathtt{ZF}}}_k}}^2}_{= \abs{\mathtt{DS}_k}^2}.\nbthis \label{eq_mean_abs_BUk2}
	\end{align*}
	Here, based on \eqref{eq_hkfmrti}, $E_0$ can be computed as:
	\begin{align*}
		E_0 &= \mean{\abs{\sqrt{\gamma_k} \hat{\vh}_k^\H \breve{\vh}_k + \sqrt{\eta_k}  \hat{\vh}_k^\H {\vv}  + \ve_k^\H (\sqrt{\gamma_k} \breve{\vh}_k + \sqrt{\eta_k} {\vv})}^2}\\
		&= \gamma_k \mean{\absshort{\hat{\vh}_k^\H \breve{\vh}_k}^2} + \eta_k \mean{\absshort{\hat{\vh}_k^\H {\vv}}^2} \\
		&\quad + \sqrt{\gamma_k \eta_k} \mean{\hat{\vh}_k^\H \breve{\vh}_k \hat{\vh}_k^\H {\vv}}\\
		&\quad + \eta_k \mean{\abs{\ve_k^\H {\vv}}^2} + \gamma_k \mean{\abs{\ve_k^\H \breve{\vh}_k}^2}, \nbthis \label{eq_E001}
	\end{align*}
	where
	\begin{align*}
		&\mean{\absshort{\hat{\vh}_k^\H \breve{\vh}_k}^2} = \mean{\Big[ \hat{\mH}^\H \hat{\mH} (\hat{\mH}^\H \hat{\mH})^{-1} \Big]_{kk}} = 1, \nbthis \label{eq_mean_1}\\
		&\mean{\absshort{\hat{\vh}_k^\H {\vv}}^2} = \mean{\normshort{\hat{\vh}_k}^2} = \Nt \xi_k, \nbthis \label{eq_mean_2}\\
		&\mean{\hat{\vh}_k^\H \breve{\vh}_k \hat{\vh}_k^\H {\vv}} = \mean{\Big[ \hat{\mH}^\H \hat{\mH} (\hat{\mH}^\H \hat{\mH})^{-1} \Big]_{kk} \hat{\vh}_k^\H {\vv}} = \bm{0}, \nbthis \label{eq_mean_3} \\
		&\mean{\abs{\ve_k^\H {\vv}}^2} = \mean{\norm{\ve_k}^2} = \Nt \epsilon_k. \nbthis \label{eq_mean_4}
	\end{align*}
	Here, \eqref{eq_mean_2}--\eqref{eq_mean_4} follow from the fact that all the elements of ${\vv}$ are deterministic with unit modulus and $\hat{\vh}_k \sim \mathcal{CN}(0, \xi_k \mI_{\Nt})$. 
	To compute $\meanshort{\absshort{\ve_k^\H \breve{\vh}_k}^2}$, we write $\breve{\vh}_k = \hat{\mH} (\hat{\mH}^\H \hat{\mH})^{-1} \vi_k$, where $\vi_k$ is the $k$th column of $\mI_K$. As a result, we have
	\begin{align*}
		&\mean{\absshort{\ve_k^\H \breve{\vh}_k}^2} = \epsilon_k \mean{\normshort{\breve{\vh}_k}^2} = \epsilon_k \mean{\tr{\breve{\vh}_k \breve{\vh}_k^\H}} \\
		&= \epsilon_k \mean{\tr{(\hat{\mH}^\H \hat{\mH})^{-1} \vi_k \vi_k^\H  (\hat{\mH}^\H \hat{\mH})^{-1} \hat{\mH}^\H \hat{\mH} }}\\
		&= \epsilon_k \mean{\Big[(\hat{\mH}^\H \hat{\mH})^{-1}\Big]_{kk}} = \frac{\epsilon_k}{(\Nt-K)\xi_k}, \nbthis \label{eq_mean_norm_h_ZF}
	\end{align*}
	where the last equality follows from Proposition 3 in~\cite{ngo2013energy}.
	From \eqref{eq_E0001}, \eqref{eq_mean_abs_BUk2}, and \eqref{eq_mean_1}--\eqref{eq_mean_4}, and noting that $\xi_k + \epsilon_k = \beta_k$, we obtain
	\begin{align*}
		\mean{\abs{\mathtt{BU}_k}^2} 
		= \Nt \beta_k \eta_k  + \frac{\epsilon_k \gamma_k}{(\Nt-K) \xi_k}. \nbthis \label{eq_mean_abs_BUk2_1}
	\end{align*}

	\subsubsection{Compute $\meanshort{\abs{\mathtt{UI}_{kj}}^2}$}
	Similar to \eqref{eq_E001}, we can write
	\begin{align*}
		&\mean{\abs{\mathtt{UI}_{kj}}^2} = \gamma_j \mean{\absshort{\hat{\vh}_k^\H \breve{\vh}_j}^2} + \eta_j \mean{\absshort{\hat{\vh}_k^\H {\vv}}^2} \\
		&\qquad + \eta_j \mean{\absshort{\ve_k^\H {\vv}}^2} + \gamma_j \mean{\absshort{\ve_k^\H \breve{\vh}_j}^2}.
	\end{align*}
    By using $\meanshort{\absshort{\hat{\vh}_k^\H \breve{\vh}_j}^2} = \meanshort{[ \hat{\mH}^\H \hat{\mH} (\hat{\mH}^\H \hat{\mH})^{-1} ]_{kj}} = 0$ and the results in \eqref{eq_mean_1}--\eqref{eq_mean_4}, we obtain
    \begin{align*}
    	\mean{\abs{\mathtt{UI}_{kj}}^2} 
    	&= \Nt \beta_k \eta_j + \frac{\epsilon_k \gamma_j }{(\Nt-K) \xi_j}. \nbthis \label{eq_mean_abs_UIk2}
    \end{align*}
    From \eqref{eq_mean_abs_BUk2_1} and \eqref{eq_mean_abs_UIk2}, we have
    \begin{align*}
    	\mean{\abs{\mathtt{BU}_k}^2} +  \sum\nolimits_{j \neq k} \mean{\abs{\mathtt{UI}_{kj}}^2}
    	&= \Nt \beta_k \rho + \Nt {\bm{\zeta}_{\mathtt{ZF}}^\T}_k \bm{\gamma},
    \end{align*}
    where ${\bm{\zeta}_{\mathtt{ZF}}}_k$ is defined in Theorem \ref{theo_SE}. Then, we obtain
    \begin{align*}
    	{\mathtt{SINR}_{\mathtt{ZF}}}_k = \frac{\gamma_k}{\Nt \beta_k \rho +  \Nt {\bm{\zeta}_{\mathtt{ZF}}^\T}_k \bm{\gamma} + \sigmac},
    \end{align*}
    leading to \eqref{eq_SE_theo} with $\lambdabf_k = {\lambda_{\mathtt{ZF}}}_k$ and $\zetabf_k = {\bm{\zeta}_{\mathtt{ZF}}}_k$.

	\section{Proof of Theorem \ref{theo_CRB}}
	\label{appd_CRB}
	{The equalities in \eqref{eq_CRB_theta_0} and \eqref{eq_CRB_phi_0} follow directly from the fact that the CRLBs for $\theta$ and $\phi$ correspond to the first and second diagonal entries of $\mJ_{\theta, \phi}$ in \eqref{eq_FIM}, i.e., $\widetilde{\mathtt{CRLB}}_{\mathtt{bf},\theta}(\bm{\gamma},\rho) = [\mJ_{\theta, \phi}^{-1}]_{11}$ and $\widetilde{\mathtt{CRLB}}_{\mathtt{bf},\phi}(\bm{\gamma},\rho) =  [\mJ_{\theta, \phi}^{-1}]_{22}$, where, $[\mA]_{ij}$ denotes the $(i,j)$-th entry of matrix $\mA$. In the following, we prove the closed form expressions for $\Jtt$, $\Jpp$, $\Jtp$, $\Jaa$, and $\Jpa$ given in \eqref{eq_Jtt}--\eqref{eq_Jpa}, respectively.}

    We first derive common terms for the CRLBs for both the MRT and ZF cases. Specifically, based on \eqref{eq_at}--\eqref{eq_atv} and using the property $\frac{\partial (\va \otimes \vb)}{\partial o} = \frac{\partial \va}{\partial o}  \otimes \vb + \va  \otimes \frac{\partial \vb}{\partial o}$, we have
	\begin{align*}
		\adottheta &= \dot{\va}_{\mathtt{h}\theta} \otimes  \av + \ah \otimes \dot{\va}_{\mathtt{v}\theta} = \dot{\va}_{\mathtt{h}\theta} \otimes  \av, \nbthis \label{eq_adot_theta} \\
		\adotphi &= \dot{\va}_{\mathtt{h}\phi} \otimes  \av + \ah \otimes \dot{\va}_{\mathtt{v}\phi}, \nbthis \label{eq_adot_phi}
	\end{align*}
	where $\dot{\va}_{\mathtt{v}\theta} = 0$ as it is independent of $\theta$, and
	\begin{align*}
		\dot{\va}_{\mathtt{h}\theta} &= j\pi \cos(\theta) \sin(\phi)  \vu_{\mathtt{th}} \circ \ah, \nbthis \label{eq_ahdot_theta} \\
		\dot{\va}_{\mathtt{h}\phi} &= j\pi \sin(\theta) \cos(\phi) \vu_{\mathtt{th}} \circ \ah, \nbthis \label{eq_ahdot_phi}\\
		\dot{\va}_{\mathtt{v}\phi} &= -j\pi \sin(\phi) \vu_{\mathtt{tv}} \circ \av. \nbthis \label{eq_avdot_phi}
	\end{align*}
	Here, $\vu_{\mathtt{x}} \triangleq \left[-(N_\mathtt{x}-1)/2, \ldots, 0, \ldots, (N_\mathtt{x}-1)/2 \right]^\T$, $\mathtt{x} \in \{\mathtt{th},\mathtt{tv},\mathtt{rh},\mathtt{rv}\}$, and
	\begin{align*}
		\norm{\vu_{\mathtt{x}}}^2 &= 2\left( \frac{(N_{\mathtt{x}}-1)^2}{4} + \ldots + 1 \right) = \frac{N_{\mathtt{x}}(N_{\mathtt{x}}^2-1)}{12}. \nbthis \label{eq_norm_ux}
	\end{align*}
	Note that $\ah^\H \dot{\va}_{\mathtt{h}\theta} = \ah^\H \dot{\va}_{\mathtt{h}\phi} = 0$, and since $\left(\vz_1 \otimes \vz_2 \right)^\H = \vz_1^\H \otimes \vz_2^\H$ and $\left(\vz_1 \otimes \vz_2 \right)\left(\vz_3 \otimes \vz_4 \right) = \left(\vz_1 \vz_3 \right) \otimes \left(\vz_2 \vz_4 \right)$, we have
	\begin{align*}
		\va^\H  \adottheta &= \left(\ah \otimes \av\right)^\H \left(\dot{\va}_{\mathtt{h}\theta} \otimes  \av\right)
		= \left(\ah^\H \dot{\va}_{\mathtt{h}\theta} \otimes \av^\H \av\right) = 0. \nbthis \label{eq_aadotheta}\\
		\va^\H  \adotphi &= \left(\ah \otimes \av\right)^\H \left(\dot{\va}_{\mathtt{h}\phi} \otimes  \av + \ah \otimes \dot{\va}_{\mathtt{v}\phi}\right)\\ 
		&= \left(\ah^\H \dot{\va}_{\mathtt{h}\phi} \otimes \av^\H \av\right) + \left(\ah^\H \ah \otimes \av^\H \dot{\va}_{\mathtt{v}\phi} \right) = 0. \nbthis \label{eq_aadotphi}
	\end{align*}
	Similarly, we can show that
	\begin{align*}
		\vb^\H  \bdottheta = 0,\
		\vb^\H  \bdotphi = 0. \nbthis \label{eq_bbdot}
	\end{align*}

    Next, we derive matrix $\mR_{\mathtt{bf}}$ in \eqref{eq_approx_cov}. For the MRT precoder, from \eqref{eq_F_MRT} and \eqref{eq_approx_cov}, we have
	\begin{align*}
		&\mR_{\mathtt{MRT}} 
		= \mean{\hat{\mH} \bm{\Gamma}^2 \hat{\mH}^\H + \rho {\vv}  {\vv}^\H + {\vv} \bar{\bm{\eta}}^\T \bm{\Gamma}^\H \hat{\mH}^\H + \hat{\mH} \bm{\Gamma} \bar{\bm{\eta}} {\vv}^\H} \\
		&\quad= \sum_{k=1}^K \gamma_k \mean{\hat{\vh}_k \hat{\vh}_k^\H} + \rho {\vv}  {\vv}^\H = \bm{\xi}_{\mathtt{MRT}}^\T \bm{\gamma} \mI_{\Nt} + \rho {\vv}  {\vv}^\H.\nbthis \label{eq_Rx_mrt}
	\end{align*}
    Similarly, for the ZF precoder, we have
	\begin{align*}
		\mR_{\mathtt{ZF}} 
		&= \underbrace{\mean{\hat{\mH} (\hat{\mH}^\H \hat{\mH})^{-1} \bm{\Gamma}^2 (\hat{\mH} (\hat{\mH}^\H \hat{\mH})^{-1})^\H}}_{\triangleq \mA} + \rho {\vv}  {\vv}^\H.
	\end{align*}
	To compute $\mA$, recall that we can express $\hat{\mH}$ as $\hat{\mH} = \mZ \bm{\Xi}^{\frac{1}{2}}$, where $\mZ = [\vz_1, \ldots, \vz_k]$, $\vz_k \sim \mathcal{CN}(\bm{0}, \mI_{\Nt})$, and $\bm{\Xi} = \diag{\xi_1, \ldots, \xi_K}$. Then, we can write
	\begin{align*}
		\mA &= \mean{\mZ \bm{\Xi}^{\frac{1}{2}} ( \bm{\Xi}^{\frac{1}{2}} \mZ^\H \mZ \bm{\Xi}^{\frac{1}{2}})^{-1} \bm{\Gamma}^2 ( \bm{\Xi}^{\frac{1}{2}} \mZ^\H \mZ \bm{\Xi}^{\frac{1}{2}})^{-1} \bm{\Xi}^{\frac{1}{2}} \mZ^\H} \\
		&= \mean{\mZ (\mZ^\H \mZ)^{-1} \bm{\Xi}^{-\frac{1}{2}} \bm{\Gamma}^2 \bm{\Xi}^{-\frac{1}{2}} (\mZ^\H \mZ)^{-1} \mZ^\H}. \nbthis \label{eq_e2_1}
	\end{align*}
	For any  $\Nt \times \Nt$ unitary matrix $\mU$, we have
	\begin{align*}
		\mU \mA \mU^\H = \mean{\mU \mZ (\mZ^\H \mZ)^{-1} \bm{\Xi}^{-\frac{1}{2}} \bm{\Gamma}^2 \bm{\Xi}^{-\frac{1}{2}} (\mZ^\H \mZ)^{-1} (\mU \mZ)^\H}.
	\end{align*}
	Let $\bar{\mZ} = \mU \mZ$. Then, $\bar{\mZ}$ is statistically identical to $\mZ$ and $\mZ^\H \mZ = \mZ^\H \mU^\H \mU \mZ = \bar{\mZ}^\H \bar{\mZ}$. Thus, 
	\begin{align*}
		\mU \mA \mU^\H = \mean{\bar{\mZ} (\bar{\mZ}^\H \bar{\mZ})^{-1} \bm{\Xi}^{-\frac{1}{2}} \bm{\Gamma}^2 \bm{\Xi}^{-\frac{1}{2}} (\bar{\mZ}^\H \bar{\mZ})^{-1} \bar{\mZ}^\H} = \mA.
	\end{align*}
	Let $\mA = \mV \bm{\Lambda} \mV^\H$ be the eigenvalue decomposition of $\mA$. Then $\mU \mV \bm{\Lambda} \mV^\H \mU^\H = \mA$, implying that any set of mutually orthogonal vectors constitutes valid eigenvectors for $\mA$. Thus, $\mA$ must be a scaled identity, i.e., $\mA = c \mI_{\Nt}$, where $c = \frac{1}{\Nt(\Nt-K)} \sum_{k=1}^K \frac{\gamma_k}{\xi_k} = \bm{\xi}_{\mathtt{ZF}}^\T \bm{\gamma}$, with $\bm{\xi}_{\mathtt{ZF}}$ defined in Lemma \ref{lemma_power}. As a result, we can write
    \begin{align*}
        \mR_{\mathtt{ZF}} = \bm{\xi}_{\mathtt{ZF}}^\T \bm{\gamma} \mI_{\Nt}  + \rho {\vv}  {\vv}^\H. \nbthis \label{eq_Rx_ZF}
    \end{align*} 
    
    It is observed from \eqref{eq_Rx_mrt} and \eqref{eq_Rx_ZF} that $\mR_{\mathtt{ZF}}$ and $\mR_{\mathtt{MRT}}$ have a similar structure: $\mR_{\mathtt{bf}} = \bm{\xi}_{\mathtt{bf}}^\T \bm{\gamma} \mI_{\Nt}  + \rho {\vv}  {\vv}^\H$, with $\mathtt{bf} \in \{\mathtt{MRT}, \mathtt{ZF}\}$. Therefore, $\Jtt, \Jpp, \Jtp, \Jaa,$ and $\Jpat$ for both precoders can be derived in a similar manner and have similar forms. Specifically, from \eqref{eq_def_Jtt}, \eqref{eq_dot_G}, \eqref{eq_aadotheta}--\eqref{eq_Rx_mrt}, \eqref{eq_Rx_ZF}, and noting that $\norm{\va}^2 = \Nt$ and $\norm{\vb}^2 = \Nr$, we have
	\begin{align*}
		\Jtt 
		&= \kappa \abs{\alpha}^2 \mathtt{trace} \left\{ \left(\bdottheta \va^\H + \vb \adottheta^\H\right)   \left(\bm{\xi}_{\mathtt{bf}}^\T \bm{\gamma} \mI_{\Nt} + \rho {\vv}  {\vv}^\H\right) \right.\\
		&\hspace{4.5cm} \times \left. \left(\va \bdottheta^\H + \adottheta \vb^\H\right) \right\} \\
		&= \kappa \abs{\alpha}^2 \Big(\bm{\xi}_{\mathtt{bf}}^\T \bm{\gamma} \left( \Nr \norm{\adottheta}^2 + \Nt \normshort{\bdottheta}^2\right)\\
        &\hspace{1.5cm} + \rho {\left( {\norm{\vv^\H \va}^2} \normshort{\bdottheta}^2 + \Nr {\norm{\vv^\H \adottheta}^2} \right)} \Big). \nbthis \label{eq_Jtt_1}
	\end{align*}
        By a similar manner, we obtain $\Jpp, \Jtp, \Jaa,$ and $\Jpat$ as shown in Theorem \ref{theo_CRB}.

    \section{Proof of Remark \ref{rm_CRB_0}}
    \label{appd_proof_of_remark_CRB}

    With $\vv = \va$, we obtain $\norm{\vv^\H \va} = \norm{\va}^2 = \Nt$ and $\norm{\vv^\H \adottheta} = \norm{\vv^\H \adotphi} = 0$, based on \eqref{eq_aadotheta} and \eqref{eq_aadotphi}. Substituting these into \eqref{eq_Jtt}--\eqref{eq_Jpa}, we obtain:
	\begin{align*}
		&\Jtt 
		= \kappa \abs{\alpha}^2 \left(\bm{\xi}_{\mathtt{bf}}^\T \bm{\gamma} \left( \Nr \norm{\adottheta}^2 + \Nt \normshort{\bdottheta}^2\right) + \rho \Nt^2 \normshort{\bdottheta}^2\right),\\
		&\Jpp\! =\! \kappa\! \abs{\alpha}^2 \left( \bm{\xi}_{\mathtt{bf}}^\T \bm{\gamma} \left(\Nr \norm{\adotphi}^2\! +\! \Nt \normshort{\bdotphi}^2\right) \!+ \rho \Nt^2 \normshort{\bdotphi}^2\!\right)\!,\\
		&\Jtp\! =\! \kappa \abs{\alpha}^2\!\! \left(\!\bm{\xi}_{\mathtt{bf}}^\T \bm{\gamma} \!\!\left( \Nr \adottheta^\H \adotphi  + \Nt \bdottheta^\H \bdotphi\right)\! +\! \rho \Nt^2 \bdottheta^\H \bdotphi\right)\!,\\
		&\Jaa = \kappa \left( \bm{\xi}_{\mathtt{bf}}^\T \bm{\gamma} \Nt \Nr + \rho \Nt^2 \Nr \right)  \mI_2,\\
		&\Jpa = \bm{0},\ \Jpat = \Jpp.
	\end{align*}
    Substituting these results into \eqref{eq_CRB_theta_0} and \eqref{eq_CRB_phi_0}, we obtain \eqref{eq_CRB_theta_1} and \eqref{eq_CRB_phi_1}, respectively.
    
	The results in \eqref{eq_norm_adot_theta}--\eqref{eq_bdot_theta_bdot_phi} can be obtained by using the property $\ah^\H \dot{\va}_{\mathtt{h}\theta} = \ah^\H \dot{\va}_{\mathtt{h}\phi} = \ah^\H \dot{\va}_{\mathtt{v}\theta} = \ah^\H \dot{\va}_{\mathtt{v}\phi} = 0$ and the results in \eqref{eq_adot_theta}--\eqref{eq_norm_ux}. For example, $\adottheta^\H \adotphi$ can be computed as
	\begin{align*}
		\adottheta^\H \adotphi &= \left(\dot{\va}_{\mathtt{h}\theta} \otimes  \av\right)^\H \left(\dot{\va}_{\mathtt{h}\phi} \otimes  \av + \ah \otimes \dot{\va}_{\mathtt{v}\phi}\right) \nbthis \label{atheta_adotphi_0} \\
		&= \dot{\va}_{\mathtt{h}\theta}^\H \dot{\va}_{\mathtt{h}\phi} \otimes \av^\H \av + \dot{\va}_{\mathtt{h}\theta}^\H \ah \otimes \av^\H \dot{\va}_{\mathtt{v}\phi} \nbthis \label{atheta_adotphi_1} \\
		&= \norm{\av}^2 \left(\dot{\va}_{\mathtt{h}\theta}^\H \dot{\va}_{\mathtt{h}\phi}\right) \nbthis \label{atheta_adotphi_2}   \\
		&= \Ntv \pi^2 \cos(\theta) \sin(\phi) \sin(\theta) \cos(\phi) \norm{\vu_{\mathtt{th}}}^2 \nbthis \label{atheta_adotphi_3} \\
        &= \frac{\Nt(\Nth^2-1)}{12} \pi^2 \sin(\phi) \sin(\theta) \cos(\theta) \cos(\phi), \nbthis \label{atheta_adotphi_4} 
	\end{align*}
    where the second term in \eqref{atheta_adotphi_1} is equal to zero, \eqref{atheta_adotphi_3} follows from \eqref{eq_ahdot_theta} and \eqref{eq_ahdot_phi}, and \eqref{atheta_adotphi_4} follows from \eqref{eq_norm_ux} and $\Nth \Ntv = \Nt$. The other results in \eqref{eq_norm_adot_theta}--\eqref{eq_bdot_theta_bdot_phi} are obtained similarly.

    \section{Proof of Remark \ref{rm_CRB_1}}
    \label{appd_CRB_go_to_zero}
    With $\Nth = \Ntv = \sqrt{\Nt}$ and $\Nt \gg 1$, we can make the following approximations: $\norm{\adottheta}^2 \approx \Nt^2 c_1$, $\norm{\adotphi}^2 \approx \Nt^2 c_2$, and $\adottheta^\H \adotphi \approx  \Nt^2 c_3$. Here, $c_1 \triangleq \frac{1}{12}  \pi^2 \cos^2(\theta) \sin^2(\phi)$, $c_2 \triangleq \frac{1}{12} \pi^2 \cos^2(\phi)  \left(\sin^2(\theta)  + 1\right)$, and $c_3 \triangleq \frac{1}{12} \pi^2 \sin(\phi) \sin(\theta) \cos(\phi) \cos(\theta)$. Then, the CRLBs in \eqref{eq_CRB_theta_1} and \eqref{eq_CRB_phi_1} can be approximated by $\overline{\mathtt{CRLB}}_{\mathtt{bf},\theta} \approx \frac{1}{\kappa \abs{\alpha}^2 \Nt^2 } f_{\theta}^{-1}$ and $\overline{\mathtt{CRLB}}_{\mathtt{bf},\phi} \approx \frac{1}{\kappa \abs{\alpha}^2 \Nt^2} f_{\phi}^{-1}$, respectively, where 
   \begin{align*}
        f_{\theta} &= \xibf^\T \bm{\gamma} \Nr  c_1 + \rho  \normshort{\bdottheta}^2 - \frac{\left(\xibf^\T \bm{\gamma} \Nr  c_3 + \rho  \bdottheta^\H \bdotphi\right)^2}{ \xibf^\T \bm{\gamma} \Nr  c_2 + \rho  \normshort{\bdotphi}^2},\\
        f_{\phi} &= \xibf^\T \bm{\gamma} \Nr  c_2 + \rho  \normshort{\bdotphi}^2 - \frac{\left(\xibf^\T \bm{\gamma} \Nr  c_3  + \rho  \bdottheta^\H \bdotphi\right)^2}{\xibf^\T \bm{\gamma} \Nr  c_1 + \rho  \normshort{\bdottheta}^2}.
    \end{align*}
       
            Furthermore, with the equal power allocation in \eqref{eq_equal_power_solution}, we have $\xibf^\T \bm{\gamma} = \frac{P_{\mathtt{bf}}}{2\Nt} = \rho$ for both MRT and ZF precoding. Consequently, the CRLBs can be simplified as
            \begin{align*}
                \overline{\mathtt{CRLB}}_{\theta}\! &=\! 
                \frac{2}{ \kappa \abs{\alpha}^2\! \Nt P_{\mathtt{bf}} } \!\left(\!\Nr  c_1 \!+\!  \normshort{\bdottheta}^2 - \frac{\left(\!\Nr  c_3 +  \bdottheta^\H \bdotphi\right)^2}{ \Nr  c_2 +  \normshort{\bdotphi}^2}\!\right)^{-1}\!\!\!, \\
                \overline{\mathtt{CRLB}}_{\phi}\! &=\! 
                \frac{2 }{ \kappa \abs{\alpha}^2\! \Nt P_{\mathtt{bf}}}  \!\left(\!\Nr  c_2 \!+\!  \normshort{\bdotphi}^2 - \frac{\left(\Nr  c_3  +  \bdottheta^\H \bdotphi\right)^2}{\Nr  c_1 +  \normshort{\bdottheta}^2}\!\right)^{-1}\!\!\!.
         \end{align*}
          Note that $c_1$, $c_2$, $c_3$, $\normshort{\bdottheta}^2$, $\normshort{\bdotphi}^2$, and $\bdottheta^\H \bdotphi$ do not depend on $\Nt$. Thus, we can conclude that $\overline{\mathtt{CRLB}}_{\theta} \rightarrow 0$ and $\overline{\mathtt{CRLB}}_{\phi} \rightarrow 0$ when $\Nt \rightarrow \infty$.

\bibliographystyle{IEEEtran}
\bibliography{IEEEabrv,Bibliography}

 \end{document}

%% file: Definition.tex
\usepackage{amsmath,graphicx}
\usepackage{color}
\usepackage{graphicx}
\usepackage{epstopdf}
\usepackage{amsmath}
\usepackage{amssymb}
\usepackage{mathrsfs}
\usepackage{hyperref}
\usepackage{tikz}
\usepackage{bm}
\usepackage[english]{babel}
\usepackage{cite}
\usepackage{rotfloat}
\usepackage{mathtools}
\usepackage[font=normalsize,labelfont=bf]{caption}
\usepackage{amsmath}
\usepackage{makecell}
\usepackage{algorithm,algorithmic}
\usepackage{multirow}
\usepackage{subfigure}
\usepackage{booktabs}
\usepackage{colortbl}
\usepackage{multirow}
\usepackage{hhline}
\usepackage{stfloats}
\usepackage{multicol}
\usepackage{bbm}
\usepackage{cases}
\graphicspath{ {Figures/} }
\newcommand{\T}{{\scriptscriptstyle\mathsf{T}}}
\renewcommand{\H}{{\scriptscriptstyle\mathsf{H}}}

\newsavebox{\foobox}

\definecolor{kugray5}{RGB}{224,224,224}

\usepackage[normalem]{ulem}
\newcommand\rsout{\bgroup\markoverwith
	{\textcolor{red}{\rule[0.5ex]{2pt}{0.8pt}}}\ULon}



\makeatletter
\newcommand{\ALOOP}[1]{\ALC@it\algorithmicloop\ #1%
	\begin{ALC@loop}}
	\newcommand{\ENDALOOP}{\end{ALC@loop}\ALC@it\algorithmicendloop}
\renewcommand{\algorithmicrequire}{\textbf{Input:}}

\makeatother

\usepackage{etoolbox}
\let\mybibitem\bibitem
\renewcommand{\bibitem}[1]{%
	\ifstrequal{#1}{nature}
	{\color{blue}\mybibitem{#1}}
	{\color{black}\mybibitem{#1}}%
}

\graphicspath{ {Figures/} }

\newtheorem{theorem}{\textbf{Theorem}}
\newtheorem{remark}{\textbf{Remark}}
\newtheorem{lemma}{\textbf{Lemma}}

\newtheorem{proof}{Proof}

\newcommand{\epr}{\hfill\(\Box\)}

\DeclareCaptionLabelSeparator{periodspace}{.\quad}

\renewcommand{\figurename}{Fig.}
\addto\captionsenglish{\renewcommand{\figurename}{Fig.}}
\newcommand\nbthis{\addtocounter{equation}{1}\tag{\theequation}}

\newcommand{\norm}[1]{\left\lVert#1\right\rVert} 
\newcommand{\normshort}[1]{\lVert#1\rVert} 
\newcommand{\abs}[1]{\left|#1\right|} 
\newcommand{\absshort}[1]{|#1|} 

\newcommand{\tr}[1]{\mathtt{trace}\left(#1\right)} 
\newcommand{\trshort}[1]{\mathtt{trace}(#1)} 

\newcommand{\diag}[1]{\mathtt{diag}\left\{#1\right\}} 

\newcommand{\re}[1]{\mathfrak{R}{\left(#1\right)}}
\newcommand{\im}[1]{\mathfrak{I}{\left(#1\right)}}
\allowdisplaybreaks

\newcommand{\mean}[1]{\mathbb{E} \left\{#1\right\}}
\newcommand{\meanshort}[1]{\mathbb{E} \{#1\}}

\usepackage{setspace}

\newcommand{\mR}{{\mathbf{R}}}
\newcommand{\mH}{{\mathbf{H}}} 

\newcommand{\mA}{{\mathbf{A}}}
\newcommand{\mS}{{\mathbf{S}}}
\newcommand{\mW}{{\mathbf{W}}}

\newcommand{\mI}{\textbf{\textbf{I}}}

\newcommand{\mX}{{\mathbf{X}}}
\newcommand{\mY}{{\mathbf{Y}}}
\newcommand{\mG}{{\mathbf{G}}}
\newcommand{\mF}{{\mathbf{F}}}
\newcommand{\mU}{{\mathbf{U}}}
\newcommand{\mV}{{\mathbf{V}}}

\newcommand{\mZ}{{\mathbf{Z}}}

\newcommand{\mN}{{\mathbf{N}}}

\newcommand{\mJ}{{\mathbf{J}}}

\newcommand{\setC}{\mathbb{C}} 
\newcommand{\setR}{\mathbb{R}}

\newcommand{\ve}{{\mathbf{e}}} 
\newcommand{\vs}{{\mathbf{s}}}
\newcommand{\vx}{{\mathbf{x}}}
\newcommand{\vy}{{\mathbf{y}}}

\newcommand{\vv}{{\mathbf{v}}}
\newcommand{\vn}{{\mathbf{n}}}
\newcommand{\vu}{{\mathbf{u}}}
\newcommand{\vz}{{\mathbf{z}}} 
\newcommand{\vh}{{\mathbf{h}}}

\newcommand{\vb}{{\mathbf{b}}}
\newcommand{\vw}{{\mathbf{w}}}
\newcommand{\va}{{\mathbf{a}}}

\newcommand{\vp}{{\mathbf{p}}}
\newcommand{\vi}{{\mathbf{i}}}
\newcommand{\vj}{{\mathbf{j}}}

\newcommand{\vf}{\mathbf{f}}

\def\argmin{\mathop{\mathtt{argmin}}}

\def\bGamma{{\pmb{\Gamma}}}

\def\b0{{\pmb{0}}}

\newcommand{\Nr}{N_\mathtt{r}}
\newcommand{\Nt}{N_\mathtt{t}}

\newcommand{\Nrh}{N_\mathtt{rh}}
\newcommand{\Nrv}{N_\mathtt{rv}}
\newcommand{\Nth}{N_\mathtt{th}}
\newcommand{\Ntv}{N_\mathtt{tv}}

\newcommand{\ah}{\va_\mathtt{h}} 
\newcommand{\av}{\va_\mathtt{v}}

\newcommand{\sigmac}{\sigma_{\mathtt{c}}^2}
\newcommand{\sigmas}{\sigma_{\mathtt{s}}^2}

\newcommand{\noise}{\sigma^2}
\newcommand{\Pt}{P_{\mathtt{t}}}

\newcommand{\Jtt}{J_{\theta \theta}}
\newcommand{\Jtp}{J_{\theta \phi}}
\newcommand{\Jta}{\vj_{\theta \bar{\bm{\alpha}}}}
\newcommand{\Jpp}{J_{\phi \phi}}
\newcommand{\Jpa}{\vj_{\phi \bar{\bm{\alpha}}}}
\newcommand{\Jaa}{\mJ_{\bar{\bm{\alpha}}\bar{\bm{\alpha}}}}
\newcommand{\Jpat}{\tilde{J}_{\phi \bar{\bm{\alpha}}}}

\newcommand{\xibf}{\bm{\xi}_{\mathtt{bf}}}
\newcommand{\lambdabf}{{\lambda_{\mathtt{bf}}}} 
\newcommand{\zetabf}{{\bm{\zeta}_{\mathtt{bf}}}}
\newcommand{\adotphi}{\dot{\va}_{\phi}}
\newcommand{\adottheta}{\dot{\va}_{\theta}}
\newcommand{\bdotphi}{\dot{\vb}_{\phi}}
\newcommand{\bdottheta}{\dot{\vb}_{\theta}}

%% file: main.bbl
\begin{thebibliography}{10}
\providecommand{\url}[1]{#1}
\csname url@samestyle\endcsname
\providecommand{\newblock}{\relax}
\providecommand{\bibinfo}[2]{#2}
\providecommand{\BIBentrySTDinterwordspacing}{\spaceskip=0pt\relax}
\providecommand{\BIBentryALTinterwordstretchfactor}{4}
\providecommand{\BIBentryALTinterwordspacing}{\spaceskip=\fontdimen2\font plus
\BIBentryALTinterwordstretchfactor\fontdimen3\font minus
  \fontdimen4\font\relax}
\providecommand{\BIBforeignlanguage}[2]{{%
\expandafter\ifx\csname l@#1\endcsname\relax
\typeout{** WARNING: IEEEtran.bst: No hyphenation pattern has been}%
\typeout{** loaded for the language `#1'. Using the pattern for}%
\typeout{** the default language instead.}%
\else
\language=\csname l@#1\endcsname
\fi
#2}}
\providecommand{\BIBdecl}{\relax}
\BIBdecl

\bibitem{giordani2020toward}
M.~Giordani, M.~Polese, M.~Mezzavilla, S.~Rangan, and M.~Zorzi, ``{Toward 6G
  networks: Use cases and technologies},'' \emph{{IEEE} Commun. Mag.}, vol.~58,
  no.~3, pp. 55--61, 2020.

\bibitem{zhang2021overview}
J.~A. Zhang, F.~Liu, C.~Masouros, R.~W. Heath, Z.~Feng, L.~Zheng, and
  A.~Petropulu, ``An overview of signal processing techniques for joint
  communication and radar sensing,'' \emph{{IEEE} J. Sel. Topics Signal
  Process.}, vol.~15, no.~6, pp. 1295--1315, 2021.

\bibitem{zhang2018multibeam}
J.~A. Zhang, X.~Huang, Y.~J. Guo, J.~Yuan, and R.~W. Heath, ``Multibeam for
  joint communication and radar sensing using steerable analog antenna
  arrays,'' \emph{{IEEE} Trans. Veh. Technol.}, vol.~68, no.~1, pp. 671--685,
  2018.

\bibitem{ouyang2022performance}
C.~Ouyang, Y.~Liu, and H.~Yang, ``Performance of downlink and uplink integrated
  sensing and communications {(ISAC)} systems,'' \emph{{IEEE} Wireless Commun.
  Lett.}, vol.~11, no.~9, pp. 1850--1854, 2022.

\bibitem{liu2022integrated}
F.~Liu, Y.~Cui, C.~Masouros, J.~Xu, T.~X. Han, Y.~C. Eldar, and S.~Buzzi,
  ``Integrated sensing and communications: {T}owards dual-functional wireless
  networks for {6G} and beyond,'' \emph{{IEEE} J. Sel. Areas Commun.}, 2022.

\bibitem{ma2020joint}
D.~Ma, N.~Shlezinger, T.~Huang, Y.~Liu, and Y.~C. Eldar, ``Joint
  radar-communication strategies for autonomous vehicles: {C}ombining two key
  automotive technologies,'' \emph{{IEEE} Signal Process. Mag.}, vol.~37,
  no.~4, pp. 85--97, 2020.

\bibitem{huang2020majorcom}
T.~Huang, N.~Shlezinger, X.~Xu, Y.~Liu, and Y.~C. Eldar, ``{MAJoRCom}: A
  dual-function radar communication system using index modulation,''
  \emph{{IEEE} Trans. Signal Process.}, vol.~68, pp. 3423--3438, 2020.

\bibitem{ma2021spatial}
D.~Ma, N.~Shlezinger, T.~Huang, Y.~Shavit, M.~Namer, Y.~Liu, and Y.~C. Eldar,
  ``Spatial modulation for joint radar-communications systems: Design,
  analysis, and hardware prototype,'' \emph{{IEEE} Trans. Veh. Technol.},
  vol.~70, no.~3, pp. 2283--2298, 2021.

\bibitem{ma2021frac}
D.~Ma, N.~Shlezinger, T.~Huang, Y.~Liu, and Y.~C. Eldar, ``{FRaC}: {FMCW}-based
  joint radar-communications system via index modulation,'' \emph{{IEEE} J.
  Sel. Topics Signal Process.}, vol.~15, no.~6, pp. 1348--1364, 2021.

\bibitem{Hassanien2016Dual}
A.~{Hassanien}, M.~G. {Amin}, Y.~D. {Zhang}, and F.~{Ahmad}, ``Dual-function
  radar-communications: Information embedding using sidelobe control and
  waveform diversity,'' \emph{{IEEE} Trans. Signal Process.}, vol.~64, no.~8,
  pp. 2168--2181, 2016.

\bibitem{Kumari2018IEEE80211ad}
P.~{Kumari}, J.~{Choi}, N.~{González-Prelcic}, and R.~W. {Heath}, ``{IEEE}
  802.11ad-based radar: An approach to joint vehicular communication-radar
  system,'' \emph{{IEEE} Trans. Veh. Technol.}, vol.~67, no.~4, pp. 3012--3027,
  2018.

\bibitem{liu2018mu}
F.~Liu, C.~Masouros, A.~Li, H.~Sun, and L.~Hanzo, ``{MU-MIMO communications
  with MIMO radar: From co-existence to joint transmission},'' \emph{{IEEE}
  Trans. Wireless Commun.}, vol.~17, no.~4, pp. 2755--2770, 2018.

\bibitem{liu2020joint}
X.~Liu, T.~Huang, N.~Shlezinger, Y.~Liu, J.~Zhou, and Y.~C. Eldar, ``Joint
  transmit beamforming for multiuser {MIMO} communications and {MIMO} radar,''
  \emph{{IEEE} Trans. Signal Process.}, vol.~68, pp. 3929--3944, 2020.

\bibitem{johnston2022mimo}
J.~Johnston, L.~Venturino, E.~Grossi, M.~Lops, and X.~Wang, ``{MIMO OFDM
  dual-function radar-communication under error rate and beampattern
  constraints},'' \emph{{IEEE} J. Sel. Areas Commun.}, vol.~40, no.~6, pp.
  1951--1964, 2022.

\bibitem{liu2022joint}
F.~Liu, Y.-F. Liu, C.~Masouros, A.~Li, and Y.~C. Eldar, ``A joint
  radar-communication precoding design based on {Cram{\'e}r-Rao} bound
  optimization,'' in \emph{IEEE Radar Conf.}, 2022.

\bibitem{liu2022transmit}
X.~Liu, T.~Huang, Y.~Liu, and Y.~C. Eldar, ``Transmit beamforming with fixed
  covariance for integrated {MIMO} radar and multiuser communications,'' in
  \emph{Proc. IEEE Int. Conf. Acoust., Speech, Signal Processing}, 2022, pp.
  8732--8736.

\bibitem{pritzker2022transmit}
J.~Pritzker, J.~Ward, and Y.~C. Eldar, ``Transmit precoder design approaches
  for dual-function radar-communication systems,'' \emph{arXiv preprint
  arXiv:2203.09571}, 2022.

\bibitem{wang2023transmit}
J.~Wang, Y.~Chen, and L.~Chen, ``Transmit beamforming for {MIMO} dual
  functional radar-communication with {IQI},'' \emph{{IEEE} Trans. Veh.
  Technol.}, vol.~72, no.~12, pp. 15\,732--15\,744, 2023.

\bibitem{choi2024joint}
J.~Choi, J.~Park, N.~Lee, and A.~Alkhateeb, ``Joint and robust beamforming
  framework for integrated sensing and communication systems,'' \emph{{IEEE}
  Trans. Wireless Commun.}, vol.~23, no.~11, pp. 17\,602--17\,618, 2023.

\bibitem{zhang2023isac}
T.~Zhang, G.~Li, S.~Wang, G.~Zhu, G.~Chen, and R.~Wang, ``{ISAC-accelerated
  edge intelligence: Framework, optimization, and analysis},'' \emph{{IEEE}
  Trans. Green Commun. Network.}, vol.~7, no.~1, pp. 455--468, 2023.

\bibitem{wang2023optimizing}
Y.~Wang, Z.~Yang, J.~Cui, P.~Xu, G.~Chen, T.~Q. Quek, and R.~Tafazolli,
  ``Optimizing the fairness of {STAR-RIS} and {NOMA} assisted integrated
  sensing and communication systems,'' \emph{{IEEE} Trans. Wireless Commun.},
  vol.~23, no.~6, pp. 5895--5907, 2023.

\bibitem{liu2021cramer}
F.~Liu, Y.-F. Liu, A.~Li, C.~Masouros, and Y.~C. Eldar, ``{Cram{\'e}r-Rao bound
  optimization for joint radar-communication beamforming},'' \emph{{IEEE}
  Trans. Signal Process.}, vol.~70, pp. 240--253, 2021.

\bibitem{song2023intelligent}
X.~Song, J.~Xu, F.~Liu, T.~X. Han, and Y.~C. Eldar, ``Intelligent reflecting
  surface enabled sensing: {Cram{\'e}r-Rao} bound optimization,'' \emph{{IEEE}
  Trans. Signal Process.}, vol.~71, pp. 2011--2026, 2023.

\bibitem{zhu2023cramer}
Q.~Zhu, M.~Li, R.~Liu, and Q.~Liu, ``{Cram{\'e}r-Rao} bound optimization for
  active {RIS}-empowered {ISAC} systems,'' \emph{{IEEE} Trans. Wireless
  Commun.}, vol.~23, no.~9, pp. 11\,723--11\,736, 2024.

\bibitem{ren2022fundamental}
Z.~Ren, X.~Song, Y.~Fang, L.~Qiu, and J.~Xu, ``Fundamental {CRB}-rate tradeoff
  in multi-antenna multicast channel with {ISAC},'' in \emph{Proc. IEEE Global
  Commun. Conf.}, 2022, pp. 1261--1266.

\bibitem{song2023cram}
X.~Song, X.~Qin, J.~Xu, and R.~Zhang, ``{Cram{\'e}r-Rao} bound minimization for
  {IRS}-enabled multiuser integrated sensing and communications,'' \emph{{IEEE}
  Trans. Wireless Commun.}, vol.~23, no.~8, pp. 9714--9729, 2024.

\bibitem{song2023cramer}
X.~Song, T.~X. Han, and J.~Xu, ``{Cram{\'e}r-Rao} bound minimization for
  {IRS}-enabled multiuser integrated sensing and communication with extended
  target,'' in \emph{Proc. IEEE Int. Conf. Commun.}, 2023, pp. 5725--5730.

\bibitem{wang2021joint}
X.~Wang, Z.~Fei, J.~Huang, and H.~Yu, ``Joint waveform and discrete phase shift
  design for {RIS}-assisted integrated sensing and communication system under
  {Cram{\'e}r-Rao} bound constraint,'' \emph{{IEEE} Trans. Veh. Technol.},
  vol.~71, no.~1, pp. 1004--1009, 2021.

\bibitem{nguyen2023multiuser}
N.~T. Nguyen, N.~Shlezinger, Y.~C. Eldar, and M.~Juntti, ``Multiuser {MIMO}
  wideband joint communications and sensing system with subcarrier
  allocation,'' \emph{{IEEE} Trans. Signal Process.}, vol.~71, pp. 2997--3013,
  2023.

\bibitem{nguyen2023joint}
N.~T. Nguyen, L.~V. Nguyen, N.~Shlezinger, Y.~C. Eldar, A.~L. Swindlehurst, and
  M.~Juntti, ``Joint communications and sensing hybrid beamforming design via
  deep unfolding,'' \emph{{IEEE} J. Sel. Topics Signal Process.}, vol.~18,
  no.~5, pp. 901--916, 2024.

\bibitem{li2016optimum}
B.~Li, A.~P. Petropulu, and W.~Trappe, ``{Optimum co-design for spectrum
  sharing between matrix completion based MIMO radars and a MIMO communication
  system},'' \emph{{IEEE} Trans. Signal Process.}, vol.~64, no.~17, pp.
  4562--4575, 2016.

\bibitem{liu2017robust}
F.~Liu, C.~Masouros, A.~Li, and T.~Ratnarajah, ``{Robust MIMO beamforming for
  cellular and radar coexistence},'' \emph{{IEEE} Trans. Wireless Commun.},
  vol.~6, no.~3, pp. 374--377, 2017.

\bibitem{marzetta2010noncooperative}
T.~L. Marzetta, ``Noncooperative cellular wireless with unlimited numbers of
  base station antennas,'' \emph{{IEEE} Trans. Wireless Commun.}, vol.~9,
  no.~11, pp. 3590--3600, 2010.

\bibitem{larsson2014massive}
E.~G. Larsson, O.~Edfors, F.~Tufvesson, and T.~L. Marzetta, ``Massive {MIMO}
  for next generation wireless systems,'' \emph{{IEEE} Commun. Mag.}, vol.~52,
  no.~2, pp. 186--195, 2014.

\bibitem{ngo2013energy}
H.~Q. Ngo, E.~G. Larsson, and T.~L. Marzetta, ``Energy and spectral efficiency
  of very large multiuser {MIMO} systems,'' \emph{{IEEE} Trans. Commun.},
  vol.~61, no.~4, pp. 1436--1449, 2013.

\bibitem{fortunati2023fundamental}
S.~Fortunati, F.~Lisi, A.~M.~I. Ahmed, A.~Sezgin, M.~S. Greco, and F.~Gini,
  ``Fundamental limits for {ISAC}—asymptotics in massive {MIMO} sensing
  systems,'' in \emph{Integrated Sensing and Communications}.\hskip 1em plus
  0.5em minus 0.4em\relax Springer, 2023, pp. 119--147.

\bibitem{buzzi2019using}
S.~Buzzi, C.~D’Andrea, and M.~Lops, ``{Using massive MIMO arrays for joint
  communication and sensing},'' in \emph{Proc. Annual Asilomar Conf. Signals,
  Syst., Comp.}, 2019, pp. 5--9.

\bibitem{temiz2021dual}
M.~Temiz, E.~Alsusa, and M.~W. Baidas, ``{A dual-function massive MIMO uplink
  OFDM communication and radar architecture},'' \emph{{IEEE} Trans. on Cogn.
  Commun. Netw.}, vol.~8, no.~2, pp. 750--762, 2022.

\bibitem{qi2022hybrid}
C.~Qi, W.~Ci, J.~Zhang, and X.~You, ``Hybrid beamforming for millimeter wave
  {MIMO} integrated sensing and communications,'' \emph{{IEEE} Commun. Lett.},
  vol.~26, no.~5, pp. 1136--1140, 2022.

\bibitem{wang2022partially}
X.~Wang, Z.~Fei, J.~A. Zhang, and J.~Xu, ``Partially-connected hybrid
  beamforming design for integrated sensing and communication systems,''
  vol.~70, no.~10, pp. 6648--6660, 2022.

\bibitem{liyanaarachchi2021joint}
S.~D. Liyanaarachchi, C.~B. Barneto, T.~Riihonen, M.~Heino, and M.~Valkama,
  ``Joint multi-user communication and {MIMO} radar through full-duplex hybrid
  beamforming,'' in \emph{IEEE Int. Online Symposium Joint Commun. \& Sensing
  (JC\&S)}, 2021.

\bibitem{barneto2021beamformer}
C.~B. Barneto, T.~Riihonen, S.~D. Liyanaarachchi, M.~Heino,
  N.~Gonz{\'a}lez-Prelcic, and M.~Valkama, ``Beamformer design and optimization
  for full-duplex joint communication and sensing at mm-waves,'' \emph{arXiv
  preprint arXiv:2109.05932}, 2021.

\bibitem{nguyen2023jointssp}
N.~T. Nguyen, N.~Shlezinger, K.-H. Ngo, V.-D. Nguyen, and M.~Juntti, ``Joint
  communications and sensing design for multi-carrier {MIMO} systems,'' in
  \emph{Proc. IEEE Works. on Statistical Signal Processing}.\hskip 1em plus
  0.5em minus 0.4em\relax IEEE, 2023, pp. 110--114.

\bibitem{gao2022integrated}
Z.~Gao, Z.~Wan, D.~Zheng, S.~Tan, C.~Masouros, D.~W.~K. Ng, and S.~Chen,
  ``Integrated sensing and communication with mmwave massive {MIMO: A}
  compressed sampling perspective,'' \emph{{IEEE} Trans. Wireless Commun.},
  vol.~22, no.~3, pp. 1745--1762, 2022.

\bibitem{zhang2024integrated}
R.~Zhang, L.~Cheng, S.~Wang, Y.~Lou, Y.~Gao, W.~Wu, and D.~W.~K. Ng,
  ``Integrated sensing and communication with massive {MIMO: A} unified tensor
  approach for channel and target parameter estimation,'' \emph{{IEEE} Trans.
  Wireless Commun.}, vol.~23, no.~8, pp. 8571--8587, 2024.

\bibitem{liao2023power}
B.~Liao, H.~Q. Ngo, M.~Matthaiou, and P.~J. Smith, ``Power allocation for
  massive {MIMO-ISAC} systems,'' in \emph{{IEEE} Trans. Wireless Commun.},
  vol.~23, no.~10, 2024, pp. 14\,232--14\,248.

\bibitem{topal2024multi}
O.~A. Topal, {\"O}.~T. Demir, E.~Bj{\"o}rnson, and C.~Cavdar, ``Multi-target
  integrated sensing and communications in massive {MIMO} systems,''
  \emph{{IEEE} Wireless Commun. Lett.}, vol.~14, no.~2, pp. 345--349, 2024.

\bibitem{lu2020omnidirectional}
A.-A. Lu, X.~Gao, X.~Meng, and X.-G. Xia, ``{Omnidirectional precoding for 3D
  massive MIMO with uniform planar arrays},'' \emph{{IEEE} Trans. Wireless
  Commun.}, vol.~19, no.~4, pp. 2628--2642, 2020.

\bibitem{ngo2013massive}
H.~Q. Ngo, E.~G. Larsson, and T.~L. Marzetta, ``Massive {MU-MIMO} downlink
  {TDD} systems with linear precoding and downlink pilots,'' in \emph{Proc.
  Annual Allerton Conf. Commun., Contr., Computing}, 2013, pp. 293--298.

\bibitem{xin2015area}
Y.~Xin, D.~Wang, J.~Li, H.~Zhu, J.~Wang, and X.~You, ``Area spectral efficiency
  and area energy efficiency of massive {MIMO} cellular systems,'' \emph{{IEEE}
  Trans. Veh. Technol.}, vol.~65, no.~5, pp. 3243--3254, 2015.

\bibitem{pirzadeh2018spectral}
H.~Pirzadeh and A.~L. Swindlehurst, ``Spectral efficiency of mixed-{ADC}
  massive {MIMO},'' \emph{{IEEE} Trans. Signal Process.}, vol.~66, no.~13, pp.
  3599--3613, 2018.

\bibitem{liu2016energy}
Z.~Liu, W.~Du, and D.~Sun, ``Energy and spectral efficiency tradeoff for
  massive {MIMO} systems with transmit antenna selection,'' \emph{{IEEE} Trans.
  Veh. Technol.}, vol.~66, no.~5, pp. 4453--4457, 2016.

\bibitem{zhang2016spectral}
J.~Zhang, L.~Dai, S.~Sun, and Z.~Wang, ``On the spectral efficiency of massive
  {MIMO} systems with low-resolution {ADC}s,'' \emph{{IEEE} Commun. Lett.},
  vol.~20, no.~5, pp. 842--845, 2016.

\bibitem{kamga2016spectral}
G.~N. Kamga, M.~Xia, and S.~A{\"\i}ssa, ``Spectral-efficiency analysis of
  massive {MIMO} systems in centralized and distributed schemes,'' \emph{{IEEE}
  Trans. Commun.}, vol.~64, no.~5, pp. 1930--1941, 2016.

\bibitem{bjornson2017massive}
E.~Bj{\"o}rnson, J.~Hoydis, L.~Sanguinetti \emph{et~al.}, ``Massive {MIMO}
  networks: {S}pectral, energy, and hardware efficiency,'' \emph{Found. Trends
  Signal Process.}, vol.~11, no. 3-4, pp. 154--655, 2017.

\bibitem{van2017massive}
T.~Van~Chien and E.~Bj{\"o}rnson, ``Massive {MIMO} communications,'' \emph{5G
  Mobile Commun.}, pp. 77--116, 2017.

\bibitem{mollen2016uplink}
C.~Mollen, J.~Choi, E.~G. Larsson, and R.~W. Heath, ``{Uplink performance of
  wideband massive MIMO with one-bit ADCs},'' \emph{{IEEE} Trans. Wireless
  Commun.}, vol.~16, no.~1, pp. 87--100, 2016.

\bibitem{bekkerman2006target}
I.~Bekkerman and J.~Tabrikian, ``Target detection and localization using {MIMO}
  radars and sonars,'' \emph{{IEEE} Trans. Signal Process.}, vol.~54, no.~10,
  pp. 3873--3883, 2006.

\bibitem{nguyen2022beam}
N.~T. Nguyen, J.~Kokkoniemi, and M.~Juntti, ``Beam squint effects in {THz}
  communications with {UPA and ULA: C}omparison and hybrid beamforming
  design,'' in \emph{Proc. IEEE Global Commun. Conf. Workshop}, 2022.

\bibitem{yu2016alternating}
X.~Yu, J.-C. Shen, J.~Zhang, and K.~B. Letaief, ``Alternating minimization
  algorithms for hybrid precoding in millimeter wave {MIMO} systems,''
  \emph{{IEEE} J. Sel. Topics Signal Process.}, vol.~10, no.~3, pp. 485--500,
  2016.

\bibitem{huang2022joint}
Z.~Huang, K.~Wang, A.~Liu, Y.~Cai, R.~Du, and T.~X. Han, ``Joint pilot
  optimization, target detection and channel estimation for integrated sensing
  and communication systems,'' \emph{{IEEE} Trans. Wireless Commun.}, vol.~21,
  no.~12, pp. 10\,351--10\,365, 2022.

\bibitem{Beck:JGO:10}
A.~Beck, A.~Ben-Tal, and L.~Tetruashvili, ``A sequential parametric convex
  approximation method with applications to nonconvex truss topology design
  problems,'' \emph{J. Global Optim.}, vol.~47, no.~1, pp. 29--51, May 2010.

\bibitem{NasirTWC21}
A.~A. Nasir, H.~D. Tuan, H.~H. Nguyen, M.~Debbah, and H.~V. Poor, ``Resource
  allocation and beamforming design in the short blocklength regime for
  {URLLC},'' \emph{IEEE Trans. Wireless Commun.}, vol.~20, no.~2, pp.
  1321--1335, 2021.

\bibitem{Dinh:TCOMM:2017}
V.-D. Nguyen, T.~Q. Duong, H.~D. Tuan, O.-S. Shin, and H.~V. Poor, ``Spectral
  and energy efficiencies in full-duplex wireless information and power
  transfer,'' \emph{IEEE Trans. Commun.}, vol.~65, no.~5, pp. 2220--2233, May
  2017.

\bibitem{Dinh:JSAC:Dec2017}
V.-D. Nguyen, H.~D. Tuan, T.~Q. Duong, H.~V. Poor, and O.-S. Shin, ``Precoder
  design for signal superposition in {MIMO-NOMA} multicell networks,''
  \emph{IEEE J. Select. Areas Commun.}, vol.~35, no.~12, pp. 2681--2695, Dec.
  2017.

\bibitem{Ben:2001}
A.~Ben-Tal and A.~Nemirovski, \emph{Lectures on Modern Convex
  Optimization}.\hskip 1em plus 0.5em minus 0.4em\relax Philadelphia: MPS-SIAM
  Series on Optimi., SIAM, 2001.

\bibitem{van2021reconfigurable}
T.~Van~Chien, H.~Q. Ngo, S.~Chatzinotas, M.~Di~Renzo, and B.~Ottersten,
  ``Reconfigurable intelligent surface-assisted cell-free massive {MIMO}
  systems over spatially-correlated channels,'' \emph{{IEEE} Trans. Wireless
  Commun.}, vol.~21, no.~7, pp. 5106--5128, 2021.

\bibitem{yang2013performance}
H.~Yang and T.~L. Marzetta, ``Performance of conjugate and zero-forcing
  beamforming in large-scale antenna systems,'' \emph{{IEEE} J. Sel. Areas
  Commun.}, vol.~31, no.~2, pp. 172--179, 2013.

\end{thebibliography}
